\documentclass[11pt,draftcls,onecolumn]{IEEEtran}

\usepackage{amsmath}
\usepackage{amsfonts}
\usepackage{mathrsfs}
\usepackage{enumerate}
\usepackage{mathtools}
\usepackage{amssymb}
\usepackage{graphicx}
\usepackage{color}
\usepackage{soul}

\usepackage[bottom=0.8in]{geometry}

\newtheorem{theorem}{Theorem}

\newtheorem{corollary}{Corollary}

\newtheorem{remark}{Remark}
\newcommand {\dfn} {:=}
\newcommand {\prob} {\ensuremath{\mathbb{P}}}
\newcommand {\reals} {\ensuremath{\mathbb{R}}}
\newcommand {\naturals} {\ensuremath{\mathbb{N}}}

\newcommand {\bE} {\ensuremath{\mathbb{E}}}
\newcommand {\cX} {{\cal X}}
\newcommand {\cY} {{\cal Y}}

\allowdisplaybreaks

\begin{document}
\thispagestyle{empty}
\setcounter{page}{1}
\setlength{\baselineskip}{1.15\baselineskip}

\title{\huge{Some Useful Integral Representations for Information-Theoretic Analyses}}

\author{Neri Merhav \quad Igal Sason\\[0.3cm]
The Andrew and Erna Viterbi Faculty of Electrical Engineering\\
Technion -- Israel Institute of Technology \\
Technion City, Haifa 3200003, Israel\\
E-mail: \{merhav,sason\}@ee.technion.ac.il}

\maketitle
\thispagestyle{empty}

\begin{abstract}
This work is an extension of our earlier article, where a well--known integral
representation of the logarithmic function was explored, and was accompanied
with demonstrations of its usefulness in obtaining compact, easily--calculable,
exact formulas for quantities that involve expectations of the logarithm of a
positive random variable. Here, in the same spirit, we derive an exact integral
representation (in one or two dimensions) of the moment of a non--negative
random variable, or the sum of such independent random variables, where the
moment order is a general positive non--integer real (also known as fractional moments).
The proposed formula is applied to a variety of examples with an information--theoretic
motivation, and it is shown how it facilitates their numerical evaluations.
In particular, when applied to the calculation of a moment of the sum of a
large number, $n$, of non--negative random variables, it is clear that integration
over one or two dimensions, as suggested by our proposed integral representation,
is significantly easier than the alternative of integrating over $n$ dimensions,
as needed in the direct calculation of the desired moment.\\[0.2cm]
\end{abstract}

\break

\section{Introduction}
\label{section: introduction}

In mathematical analyses associated with many problems in information
theory and related fields, one is often faced with the need to compute
expectations of logarithmic functions of composite random variables (see, e.g.,
\cite{DZWY15,EvansB88,Knessl98,LapidothM-IT03,Martinez07,MS_Ent20,RajanT15,ZadehH16}),
or moments of such random variables, whose order may be a general positive real, not
even necessarily an integer (see, e.g., \cite{Arikan96,ArikanM98-1,ArikanM98-2,
Boztas97,BracherHL_IT19,BunteL14a,Campbell65,CV2014a,CV2014b,NM18_estimation,ISSV_IT18,
IS_Ent18,Sundaresan07,Sundaresan07b}).

In the case of the logarithmic function, the common practice is
either to resort approximate evaluations, provided by
upper and lower bounds on the desired expression (for example, by using
Jensen's inequality), or to approximate the calculations by using the Taylor series
expansion of the function $\ln x$. More recently, it has become popular to use the
replica trick (see, e.g., \cite[Chapter~8]{MM09}), which is a non-rigorous,
but useful technique, borrowed from statistical physics.

In our earlier work \cite{MS_Ent20}, we have demonstrated how the following
well--known integral representation of the logarithmic function,
\begin{align}
\label{ir}
\ln x=
\int_0^\infty\left(e^{-u}-e^{-ux}\right) \;
\frac{\mathrm{d}u}{u}, \quad x>0,
\end{align}
can be useful in a variety of application areas in the field of information theory,
including both source and channel coding, as well as other aspects of this field.
To calculate the expectation, $\bE\{\ln X\}$, where $X$ is a positive random variable,
the idea is simply to invoke the integral representation \eqref{ir} and to commute
the expectation and integration operators, i.e.,
\begin{align}
\label{Elnx}
\bE\{\ln X\}=
\int_0^\infty\left(e^{-u}-\bE\{e^{-uX}\}\right) \;
\frac{\mathrm{d}u}{u},
\end{align}
thereby replacing the calculation of $\bE\{\ln X\}$ by the calculation of the
moment--generating function (MGF), $M_X(u) := \bE\{e^{uX}\}$ for all $u \leq 0$, which is
often a lot easier to express in closed form. Moreover, in frequently encountered
situations where $X$ is given by the sum of $n$ independently identically distributed
(i.i.d.) random variables, the MGF of $X$ is given by the $n$--th power of the MGF
of a single random variable in the sum that forms $X$. This reduces the dimension
of the integration from $n$ (in the original expression) to a single dimension of
the integration over $u$. Interestingly, this integral representation has
also been used in the statistical physics literature (see, e.g., \cite{EN93},
\cite[p.~140]{MM09}, \cite{SSRCM19}), but not as much as the replica trick.

In this paper, we proceed in the same spirit as in \cite{MS_Ent20}, and we extend
the scope to propose an integral representation of a general moment of a non--negative
random variable, $X$, namely, the expectation, $\bE\{X^\rho\}$ for a given real $\rho > 0$.
Obviously, when $\rho$ is integer, this moment is simply given by the $\rho$--th order
derivative of the MGF of $X$, calculated at the origin, as is very well known. However,
the integral representation we propose, in this work, applies to any non--integer,
positive $\rho$, and here too, it replaces the direct calculation of $\bE\{X^\rho\}$
by integration of an expression that involves the MGF of $X$. We refer to this
representation as an {\em{extension}} of \eqref{Elnx}, as the latter can be obtained
as a special case of the formula for $\bE\{X^\rho\}$, by invoking one of the equivalent
identities
\begin{align}
\bE\{\ln X\} =\lim_{\rho\to 0}\frac{\bE\{X^\rho\}-1}{\rho}, \qquad
\bE\{\ln X\} =\lim_{\rho\to 0}\frac{\ln[\bE\{X^\rho\}]}{\rho}.
\end{align}
While the proposed integral representation of $\bE\{X^\rho\}$
can be readily obtained from \cite[p.~363, Identity~(3.434.1)]{GR14}
in the range $\rho\in(0,1)$, the non--trivial extension we propose
for a non--integer and real $\rho > 1$ is new to the best of our knowledge.

As in \cite{MS_Ent20}, the proposed integral representation is applied to a
variety of examples with an information--theoretic motivation, and it is
shown how it facilitates the numerical evaluations. In particular, similarly
as in the case of the logarithmic function, when applied to the calculation
of a moment of the sum of a large number, $n$, of non--negative random variables,
it is clear that integration over one or two dimensions, as suggested by our
proposed integral representation, is significantly easier than the alternative
of integrating over $n$ dimensions, as needed in the direct calculation of the
desired moment. Furthermore, single or double-dimensional integrals can be
instantly and accurately calculated using built-in numerical integration procedures.

Fractional moments have been considered in the mathematical literature
(see, e.g., \cite{Gzyl1, Gzyl2, Tagliani1, Tagliani2}).
A relationship between fractional and integer--order moments was considered
in \cite{Gzyl1} by expressing a fractional moment as an infinite
series which depends on all the positive integer--order moments, which was
followed by an algorithm for numerical calculations of fractional moments.

The outline of the remaining part of this paper is as follows. In Section
\ref{section: moments}, we provide the mathematical background associated
with the integral representation in general. In Section~\ref{section: applications},
we demonstrate this integral representation in applications, including:
moments of guesswork, moments of estimation errors, differential R\'enyi
entropies of generalized multivariate Cauchy distributions, and mutual
information calculations of a certain model of a jammed channel. Each one of
these examples occupies one subsection of Section~\ref{section: applications}.
The integral representations in this paper are not limited to the examples
in Section~\ref{section: applications}, and such representations can be proved
useful in other information--theoretic problems (see, e.g., \cite{MS_Ent20}
and references therein).

\section{Statistical Moments of Arbitrary Positive Orders}
\label{section: moments}

It is well known that any integer--order moment of a random variable
$X$ can be calculated from its MGF
\begin{align}
\label{eq: MGF}
M_X(u) := \bE\bigl\{e^{uX}\bigr\}, \quad u \in \reals,
\end{align}
by using its $\rho$--th order derivative, calculated
at $u=0$, i.e.,
\begin{align}
\label{well-known}
\bE\{X^\rho\} = M_X^{(\rho)}(0), \quad \rho \in \naturals.
\end{align}
Quite often, however, there is a theoretical and practical interest to
calculate fractional moments of non-negative random variables.
We next obtain {a closed--form integral expression of the $\rho$--th
moment of a non-negative random variable $X$, as a functional of its
MGF, for any positive real $\rho$. Before we proceed, it should be noted
that for $\rho \in (0,1)$, such an expression is available in
handbooks of standard tables of integrals, for example, in
\cite[p.~363, Identity~(3.434.1)]{GR14}. The first innovation here,
however, is in a non--trivial extension of this formula for all $\rho > 0$
as an expression that involves a one--dimensional integral. It should
be noted that although the definition of a fractional moment of a RV
is also given by a one--dimensional integral (or a sum, depending on whether
the RV is discrete or continuous), the utility of our formula is, e.g.,
in expressing the $\rho$-th moment of a sum of non--negative and independent
random variables as a one--dimensional integral, instead of an $n$--dimensional
integral which is obtained by the direct definition. This new formula serves
as the basic building block in all our information-theoretic applications
throughout this paper.

We first define the Beta and Gamma functions (see, e.g.,
\cite[Section~8.3]{GR14} and \cite[Chapter~5]{OlverLBC10}):
\begin{align}
\label{eq1: Beta}
B(u,v) &\dfn \int_0^1 t^{u-1} (1-t)^{v-1} \, \mathrm{d}t, \quad u,v > 0, \\[0.1cm]
\label{eq: Gamma}
\Gamma(u) &\dfn \int_0^\infty t^{u-1} e^{-t} \, \mathrm{d}t, \quad u > 0,
\end{align}
where these functions are related by the equality
\begin{align}
\label{eq2: Beta}
B(u,v) = \frac{\Gamma(u) \, \Gamma(v)}{\Gamma(u+v)}, \quad u,v > 0.
\end{align}
\begin{theorem}
\label{thm: rho moment}
Let $X$ be a non-negative random variable
with an MGF $M_X(\cdot)$, and let
$\rho > 0$ be a non-integer real. Then,
\begin{align}
\bE\{X^\rho\} &= \frac1{1+\rho} \, \sum_{\ell=0}^{\lfloor \rho \rfloor}
\frac{\alpha_\ell}{B(\ell+1,\rho+1-\ell)} \nonumber \\
\label{eq: rho>0 non-int.}
&\hspace*{0.4cm} + \frac{\rho \, \sin(\pi \rho) \, \Gamma(\rho)}{\pi}
\int_0^\infty \frac1{u^{\rho+1}} \,
\Biggl( \, \sum_{j=0}^{\lfloor \rho \rfloor}
\biggl \{ \frac{(-1)^j \, \alpha_j}{j!} \; u^j \biggr\} \, e^{-u} - M_X(-u) \Biggr)
\, \mathrm{d}u,
\end{align}
where for all $j \in \{0, 1, \ldots, \}$
\begin{align}
\label{eq: alpha_k}
\alpha_j &\dfn \bE\bigl\{(X-1)^j\bigr\} \\
\label{eq2: alpha_k}
& \hspace*{0.1cm} = \frac1{j+1} \sum_{\ell=0}^j \frac{(-1)^{j-\ell} \,
M_X^{(\ell)}(0)}{B(\ell+1,j-\ell+1)}.
\end{align}
\end{theorem}

\begin{proof}
See Appendix~\ref{appendix: theorem - rho moment}.
\end{proof}

\begin{remark}
The proof of \eqref{eq: rho>0 non-int.} in Appendix~\ref{appendix: theorem - rho moment}
does not apply to $\rho \in \naturals$ (see \eqref{0120c}, \eqref{eq: c0} etc., where the
denominators vanish for $\rho \in \naturals$). In the latter case, by referring to the second
term on the right--hand side of \eqref{eq: rho>0 non-int.}, we get $\sin(\pi \rho) = 0$
and also the integral diverges (specifically, for $\rho \in \naturals$, the integrand
scales like $\frac1u$ for $u$ that is sufficiently close to zero), yielding an expression
of the type $0 \cdot \infty$. However, taking a limit in \eqref{eq: rho>0 non-int.}
where we let $\rho$ tend to an integer, and applying L'H\^{o}pital's rule can reproduce the
well--known result in \eqref{well-known}.
\end{remark}

\begin{corollary}
\label{corollary: rho moment}
For any $\rho\in(0,1)$,
\begin{align}
\label{eq: rho-in-(0,1)}
\bE\{X^\rho\} &= 1 + \frac{\rho}{\Gamma(1-\rho)} \int_0^{\infty}
\frac{e^{-u} - M_X(-u)}{u^{1+\rho}} \; \mathrm{d}u.
\end{align}
\end{corollary}

\begin{proof}
Eq.~\eqref{eq: rho-in-(0,1)} is due to Theorem~\ref{thm: rho moment},
and by using \eqref{eq: Gamma recursion}, \eqref{eq: Gamma identity 1}
(see Appendix~\ref{appendix: theorem - rho moment}) and $\alpha_0 \dfn 1$, which give
\begin{align}
& \Gamma(\rho) \, \Gamma(1-\rho) = \frac{\pi}{\sin(\pi \rho)}, \\
\label{eq: Betta identity 1}
& \frac1{1+\rho} \, \frac{\alpha_0}{B(1,\rho+1)}
= \frac1{1+\rho} \, \frac{\Gamma(\rho+2)}{\Gamma(\rho+1)} = 1.
\end{align}
\end{proof}

\begin{remark}
Corollary~\ref{corollary: rho moment} also follows from \cite[p.~363, Identity~(3.434.1)]{GR14}
(see \cite[Section~4]{MS_Ent20}).
\end{remark}

\begin{corollary} \cite{MS_Ent20}
\label{corollary: log. expectation}
Let $X$ be a positive random variable. Then,
\begin{align}
\label{log. expectation}
\bE\{\ln X\} = \int_0^{\infty} \frac{e^{-u} - M_X(-u)}{u} \; \mathrm{d}u.
\end{align}
\end{corollary}

A proof of \eqref{log. expectation} is presented in \cite[Section~2]{MS_Ent20}, based
on the integral representation of the logarithmic function in \eqref{ir},
and by interchanging the integration and the expectation. It can be alternatively
proved by using Corollary~\ref{corollary: rho moment}, and the identity
$\ln x = \underset{\rho \to 0}{\lim} \frac{x^\rho-1}{\rho}$ for $x>0$.
Identity~\eqref{log. expectation} has many useful information--theoretic applications
on its own right, as demonstrated in \cite{MS_Ent20}, and here we add even some more.
The current work is an extension and further development of \cite{MS_Ent20}, whose main
theme is in exploiting Theorem \ref{thm: rho moment} and studying its information--theoretic
applications, as well as some more applications of the logarithmic expectation.

\section{Applications}
\label{section: applications}

In this section, we exemplify the usefulness of the integral representation
of the $\rho$-th moment in Theorem~\ref{thm: rho moment} and the logarithmic
expectation in several problem areas in information theory and statistics.
These include analyses of randomized guessing, estimation errors, R\'enyi
entropy of $n$-dimensional generalized Cauchy distributions, and finally,
calculations of the mutual information for channels with a certain jammer model.
To demonstrate the direct computability of the relevant quantities, we also
present graphs of their numerical calculations.

\subsection{Moments of Guesswork}
\label{subsec: randomaized guessing}

Consider the problem of guessing the realization of a random variable
which takes on values in a finite alphabet, using a sequence of
yes/no questions of the form ``Is $X=x_1$?'', ``Is $X=x_2$?'', etc.,
until a positive response is provided by a party that observes the actual
realization of $X$. Given a distribution of $X$, a commonly used performance
metric for this problem is the expected number of guesses or, more generally,
the $\rho$-th moment of the number of guesses until $X$ is guessed successfully.
When it comes to guessing random vectors, say, of length $n$,
minimizing the moments of the number of guesses by different (deterministic or randomized)
guessing strategies has several applications and motivations in information theory,
such as sequential decoding, guessing passwords, etc., and it is also strongly
related to lossless source coding (see, e.g., \cite{Arikan96,ArikanM98-1,ArikanM98-2,
Boztas97,BracherHL_IT19,HanawalS11,MC20,MM19,ISSV_IT18,IS_Ent18,Sundaresan07,Sundaresan07b}).
In this vector case, the moments of the number of guesses behave as exponential functions
of the vector dimension, $n$, at least asymptotically, as $n$ grows without bound.
For random vectors with i.i.d.\ components, the best achievable asymptotic exponent
of the $\rho$-th guessing moment is expressed in \cite{Arikan96} by using the R\'{e}nyi
entropy of $X$ of order $\widetilde{\rho} := \frac1{1+\rho}$. Arikan assumed in
\cite{Arikan96} that the distribution of $X$ is known, and analyzed the optimal deterministic
guessing strategy, which orders the guesses according to non--increasing probabilities.
Refinements of the exponential bounds in \cite{Arikan96} with tight upper and
lower bounds on the guessing moments for optimal deterministic guessing were
recently derived in \cite{ISSV_IT18}. In the sequel, we refer to randomized
guessing strategies, rather than deterministic strategies, and we aim to derive
exact, calculable expressions for their associated guessing moments (as it is
later explained in this subsection).

Let the random variable $X$ take on values in a finite alphabet $\cX$. Consider
a random guessing strategy where the guesser sequentially submits a sequence of
independently drawn random guesses according to a certain probability distribution,
$\widetilde{P}(\cdot)$, defined on $\cX$. Randomized guessing strategies have the
advantage that they can be used by multiple asynchronous agents
which submit their guesses concurrently (see \cite{MC20} and \cite{MM19}).

In this subsection, we consider the setting of randomized guessing, and obtain an
exact representation of the guessing moment in the form of a one--dimensional
integral. Let $x \in \cX$ be any realization of $X$ and let the guessing distribution,
$\widetilde{P}$, be given. The random number, $G$, of independent guesses
until success has a geometric distribution:
\begin{align}
\label{20200115a1}
\mbox{Pr}\{G=k|x\} = \bigl[1 - \widetilde{P}(x) \bigr]^{k-1} \, \widetilde{P}(x),
\quad k \in \naturals,
\end{align}
and so, the corresponding MGF is equal to
\begin{align}
M_G(u|x) &= \sum_{k=1}^{\infty} e^{ku} \, \mbox{Pr}\{G=k|x\} \nonumber \\
\label{20200115a2}
&= \frac{\widetilde{P}(x)}{e^{-u} - \bigl(1-\widetilde{P}(x)\bigr)},
\quad u < \ln \frac1{1-\widetilde{P}(x)}.
\end{align}
In view of \eqref{eq: rho>0 non-int.}--\eqref{eq2: alpha_k} and \eqref{20200115a2},
for $x \in \cX$ and non--integer $\rho>0$,
\begin{align}
\label{20200115a3}
\bE\{G^\rho | x\} &= \frac1{1+\rho} \, \sum_{\ell=0}^{\lfloor \rho \rfloor}
\frac{\alpha_\ell}{B(\ell+1,\rho+1-\ell)} \\
\nonumber
&\hspace*{0.4cm} + \frac{\rho \, \sin(\pi \rho) \, \Gamma(\rho)}{\pi}
\int_0^\infty \frac1{u^{\rho+1}} \,
\Biggl( \, \sum_{j=0}^{\lfloor \rho \rfloor}
\biggl \{ \frac{(-1)^j \, \alpha_j}{j!} \; u^j \biggr\} \,
e^{-u} - \frac{\widetilde{P}(x)}{e^u - \bigl(1-\widetilde{P}(x)\bigr)}  \Biggr)
\, \mathrm{d}u,
\end{align}
with $\alpha_0 \dfn 1$, and for $j \in \naturals$
\begin{align}
\alpha_j &\dfn \bE\bigl\{ \bigl(G-1 \bigr)^j | X=x \bigr\} \nonumber \\
&= \sum_{k=1}^\infty (k-1)^j \, \bigl(1-\widetilde{P}(x)\bigr)^{k-1}
\, \widetilde{P}(x) \nonumber \\
\label{20200115a4}
&= \widetilde{P}(x) \, \mathrm{Li}_{-j}\bigl(1-\widetilde{P}(x)\bigr).
\end{align}
In \eqref{20200115a4} $\mathrm{Li}_{-j}(\cdot)$ is a polylogarithm (see, e.g.,
\cite[Section~25.12]{OlverLBC10}), which is given by
\begin{align}
\label{eq: polylogarithm}
\mathrm{Li}_{-j}(x) = \Bigl(x \, \frac{\mathrm{d}}{\mathrm{d}x} \Bigr)^j
\, \frac{x}{1-x}, \quad \forall \, j \in \naturals \cup \{0\},
\end{align}
with $\Bigl(x \, \frac{\mathrm{d}}{\mathrm{d}x}\Bigr)^j$ denoting differentiation
with respect to $x$ and multiplication of the derivative by $x$,
repeatedly $j$ times. In particular, we have
\begin{align}
\mathrm{Li}_0(x) = \frac{x}{1-x}, \quad \mathrm{Li}_{-1}(x) = \frac{x}{(1-x)^2},
\quad \mathrm{Li}_{-2}(x) = \frac{x(1+x)}{(1-x)^3},
\end{align}
and so on.
The function $\mathrm{Li}_{-j}(x)$ is a built--in function in the Matlab and
Mathematica softwares, which is expressed as $\mathrm{polylog}(-j,x)$.
By Corollary~\ref{corollary: rho moment}, if $\rho \in (0,1)$, then
\eqref{20200115a3} is simplified to
\begin{align}
\label{20200116a1}
\bE\{G^\rho|x\} &= 1 + \frac{\rho}{\Gamma(1-\rho)} \int_0^\infty
\frac{e^{-u}-e^{-2u}}{u^{\rho+1} \bigl[\bigl(1-\widetilde{P}(x)\bigr)^{-1}
- e^{-u} \bigr]} \; \mathrm{d}u.
\end{align}
Let $P$ denote the distribution of $X$. Averaging over $X$ to get the
unconditional $\rho$--th moment using \eqref{20200116a1}, one obtains
for all $\rho \in (0,1)$,
\begin{align}
\label{20200509a1}
\bE\{G^\rho\} = 1 + \frac{\rho}{\Gamma(1-\rho)}
\int_0^1 \frac{1-z}{(-\ln z)^{\rho+1}} \sum_{x \in \cX}
\frac{P(x) \bigl(1-\widetilde{P}(x)\bigr)}{1-z\bigl(1-\widetilde{P}(x)\bigr)} \; \mathrm{d}z,
\end{align}
where \eqref{20200509a1} is obtained by using the substitution $z := e^{-u}$.
A suitable expression of such an integral is similarly
obtained, for all $\rho > 0$, by averaging \eqref{20200115a3} over $X$.
In comparison, a direct calculation of the $\rho$--th moment gives
\begin{align}
\bE\{G^\rho\} = \sum_{x \in \cX} P(x) \, \bE\{G^\rho|x\}
\label{20200509a3}
= \sum_{k=1}^{\infty} \sum_{x \in \cX} k^\rho \bigl(1-\widetilde{P}(x)\bigr)^{k-1} \, \widetilde{P}(x)\, P(x).
\end{align}
The double sum in \eqref{20200509a3} involves a numerical computation of an infinite
series, where the number of terms required to obtain a good approximation increases
with $\rho$, and needs to be determined. The right--hand side of \eqref{20200509a1},
on the other hand, involves integration over $[0,1]$. For every practical purpose,
however, definite integrals in one or two dimensions can be calculated instantly
using built-in numerical integration procedures in MATLAB, Maple, Mathematica,
or any other mathematical software tools, and the computational complexity of the
integral in \eqref{20200509a1} is not affected by $\rho$.

As a complement to \eqref{20200115a3} (which applies to a non--integral
and positive $\rho$), we obtain that the $\rho$-th moment of
the number of randomized guesses, with $\rho \in \naturals$, is equal to
\begin{align}
\bE\{G^\rho|x\} &= \bE \bigl\{ \bigl[ \bigl( G - 1 \bigr) + 1 \bigr]^\rho \, | \, x \bigr\} \nonumber \\
&= \sum_{j=0}^\rho \binom{\rho}{j} \, \bE\bigl\{ \bigl( G - 1 \bigr)^j | x \bigr\} \nonumber \\
&= \sum _{j=0}^\rho \binom{\rho}{j} \, \alpha_j \nonumber \\
\label{20200115a5}
&= 1 + \widetilde{P}(x) \sum _{j=1}^\rho \biggl\{ \binom{\rho}{j} \,
\mathrm{Li}_{-j}\bigl(1-\widetilde{P}(x)\bigr) \biggr\},
\end{align}
where \eqref{20200115a5} follows from \eqref{20200115a4} and since $\alpha_0 = 1$.
By averaging over $X$,
\begin{align}
\label{20200115a6}
\bE\{G^\rho\} = 1 + \sum_{x \in cX} \Biggl\{ P(x) \widetilde{P}(x) \sum _{j=1}^\rho \biggl\{ \binom{\rho}{j} \,
\mathrm{Li}_{-j}\bigl(1-\widetilde{P}(x)\bigr) \biggr\} \Biggr\}.
\end{align}

To conclude, \eqref{20200115a3} and its simplification in \eqref{20200116a1} for
$\rho \in (0,1)$ give calculable one--dimensional integral expressions for the $\rho$--th
guessing moment with any $\rho > 0$. This refers to a randomized guessing strategy
whose practical advantages were further explained in \cite{MC20} and \cite{MM19}.
This avoids the need of numerical calculations of infinite sums. A further simplification
for $\rho \in \naturals$ is provided in \eqref{20200115a5} and \eqref{20200115a6},
expressed in closed form as a function of polylogarithms.

\subsection{Moments of Estimation Errors}
\label{subsec: estimation}

Let $X_1, \ldots, X_n$ be i.i.d.\ random variables with an unknown
expectation $\theta$ to be estimated, and consider the simple estimator,
\begin{align}
\label{estimator}
\widehat{\theta}_n = \frac1n \sum_{i=1}^n X_i.
\end{align}
For given $\rho > 0$, we next derive an easily-calculable expression
of the $\rho$-th moment of the estimation error.

Let $D_n \dfn \bigl( \widehat{\theta}_n - \theta \bigr)^2$ and $\rho' \dfn \frac{\rho}{2}$.
By Theorem~\ref{thm: rho moment}, if $\rho>0$ is a non--integral multiple of~2, then
\begin{align}
& \bE \bigl\{ \bigl| \widehat{\theta}_n - \theta \bigr|^{\rho} \bigr\} \nonumber \\
\label{20200110a2.0}
&= \bE \bigl\{ D_n^{\rho'} \bigr\} \\
&= \frac2{2+\rho} \, \sum_{\ell=0}^{\lfloor \rho/2 \rfloor}
\frac{\alpha_\ell}{B\bigl(\ell+1, \rho/2+1-\ell \bigr)} \nonumber \\
\label{20200110a2}
&\hspace*{0.4cm} + \frac{\rho}{2\pi} \; \sin\Bigl(\frac{\pi \rho}{2}\Bigr) \, \Gamma\Bigl(\frac{\rho}{2}\Bigr)
\int_0^\infty \frac1{u^{\rho/2+1}} \,
\Biggl( \, \sum_{j=0}^{\lfloor \rho/2 \rfloor}
\biggl \{ \frac{(-1)^j \, \alpha_j}{j!} \; u^j \biggr\} \, e^{-u} - M_{D_n}(-u) \Biggr)
\, \mathrm{d}u,
\end{align}
where
\begin{align}
\label{20200110a3}
M_{D_n}(-u) = \bE \bigl\{ \exp\bigl(-u \bigl( \widehat{\theta}_n - \theta \bigr)^2 \bigr) \bigr\},
\quad \forall \, u \geq 0,
\end{align}
$\alpha_0 \dfn 1$, and for all $j \in \naturals$ (see \eqref{eq2: alpha_k})
\begin{align}
\label{20200110a4}
\alpha_j = \frac1{j+1} \sum_{\ell=0}^j \frac{(-1)^{j-\ell} \, M_{D_n}^{(\ell)}(0)}{B(\ell+1,j-\ell+1)}.
\end{align}
By Corollary~\ref{corollary: rho moment} and \eqref{20200110a2.0},
if in particular $\rho \in (0,2)$, then the right--hand side
of \eqref{20200110a2} is simplified to
\begin{align}
\label{20200110a5}
\bE \bigl\{ \bigl| \widehat{\theta}_n - \theta \bigr|^{\rho} \bigr\}
= 1 + \frac{\rho}{2 \, \Gamma(1- \frac12 \, \rho)} \int_0^\infty
u^{-(1+\frac12 \rho)} \, \bigl[ e^{-u} - M_{D_n}(-u) \bigr] \, \mathrm{d}u,
\end{align}
and, for all $k \in \naturals$,
\begin{align}
\label{20200110a6}
& \bE \bigl\{ | \widehat{\theta}_n - \theta|^{2k} \bigr\} = M_{D_n}^{(k)}(0).
\end{align}
In view of \eqref{20200110a2.0}--\eqref{20200110a6}, obtaining a closed--form expression for the
$\rho$-th moment of the estimation error, for an arbitrary $\rho > 0$, hinges on the calculation
of the right side of \eqref{20200110a3} for all $u \geq 0$.
To this end, we invoke the identity
\begin{align}
\label{eq: Gaussian identity}
e^{-uz^2} = \frac1{2 \sqrt{\pi u}} \int_{-\infty}^\infty e^{-j \omega z - \omega^2/(4u)}
\, \mathrm{d}\omega, \quad \forall \, u>0, \, z \in \reals,
\end{align}
which is the MGF of a zero--mean Gaussian random variable with variance $\frac{1}{2u}$.
Together with \eqref{20200110a3}, it gives (see Appendix~\ref{appendix: estimation}.1)
\begin{align}
\label{MGF D_n}
M_{D_n}(-u) = \frac1{2 \sqrt{\pi u}} \int_{-\infty}^\infty e^{-j \omega \theta} \,
\phi_X^{n}\Bigl(\frac{\omega}{n}\Bigr) \, e^{-\omega^2/(4u)} \, \mathrm{d}\omega,
\quad \forall \, u>0,
\end{align}
where $X$ is a generic random variable with the same distribution as of
$X_i$ for all $i$.

The combination of \eqref{20200110a2}--\eqref{20200110a6} enables to
calculate exactly the $\rho$-th moment
$\bE \bigl\{ | \widehat{\theta}_n - \theta |^{\rho} \bigr\}$,
for any given $\rho > 0$, in
terms of a two-dimensional integral.
Combining \eqref{20200110a5} and \eqref{MGF D_n} yields, for all $\rho \in (0,2)$,
\begin{align}
& \bE \bigl\{ \bigl| \widehat{\theta}_n - \theta \bigr|^{\rho} \bigr\} \nonumber \\
\label{20200111a1}
&= 1 + \frac{\rho}{2 \, \Gamma(1- \frac12 \, \rho)} \int_0^\infty
\int_{-\infty}^\infty u^{-(\rho/2+1)} \Bigl[ \tfrac12
\, e^{-u-|\omega|} - \frac1{2 \sqrt{\pi u}} \; \phi_X^{\, n}\Bigl(\frac{\omega}{n}\Bigr)
\, e^{-j \omega \theta - \omega^2/(4u)} \Bigr] \, \mathrm{d}\omega \, \mathrm{d}u,
\end{align}
where we have used the identity
$\int_{-\infty}^\infty \tfrac12 \, e^{-|\omega|} \, \mathrm{d}\omega = 1$ in the derivation
of the first term of the integral on the right--hand side of \eqref{20200111a1}.

As an example, consider the case where $\{X_i\}_{i=1}^n$ are
i.i.d.\ Bernoulli random variables with
\begin{align}
\label{PMF Bernoulli}
\prob\{X_1=1\} = \theta, \quad \prob\{X_1=0\}=1-\theta
\end{align}
where the characteristic function is given by
\begin{align}
\label{CF Bernoulli}
\phi_X(u) \dfn \bE \bigl\{ e^{juX} \bigr\} = 1 + \theta \bigl(e^{ju}-1\bigr), \quad u \in \reals.
\end{align}
Thanks to the availability of the exact expression, we can next compare the exact
$\rho$-th moment of the estimation error $|\widehat{\theta}_n - \theta|$, with the
following closed--form upper bound (see Appendix~\ref{appendix: estimation}.2)
and thereby assess its tightness:
\begin{align}
\label{UB - estimation}
\bE \bigl\{ \bigl| \widehat{\theta}_n - \theta \bigr|^{\rho} \bigr\}
\leq K(\rho, \theta) \cdot n^{-\rho/2},
\end{align}
which holds for all $n \in \naturals$, $\rho > 0$ and $\theta \in [0,1]$,
with
\begin{align}
\label{K - UB}
K(\rho, \theta) \dfn \rho \; \Gamma\Bigl(\frac{\rho}{2}\Bigr) \;
\bigl(2 \theta \, (1-\theta) \bigr)^{\rho/2}.
\end{align}

Figures~\ref{Fig1: estimation err} and~\ref{Fig2: estimation err}
display plots of $\bE \, \bigl| \widehat{\theta}_n - \theta \bigr|$
as a function of $\theta$ and $n$, in comparison to the upper bound
\eqref{UB - estimation}. The difference in the plot of
Figure~\ref{Fig1: estimation err} is significant except for the
boundaries of the interval $[0,1]$, where both the exact value
and the bound vanish. Figure~\ref{Fig2: estimation err} indicates
that the exact value of $\bE \, \bigl| \widehat{\theta}_n - \theta \bigr|$,
for large $n$, scales like $\sqrt{n}$; this is reflected from the
apparent parallelism of the curves in both graphs, and by the upper
bound \eqref{UB - estimation}.

\begin{figure}[hbt]
\centering
\vspace*{-4.7cm}
\includegraphics[width=11cm]{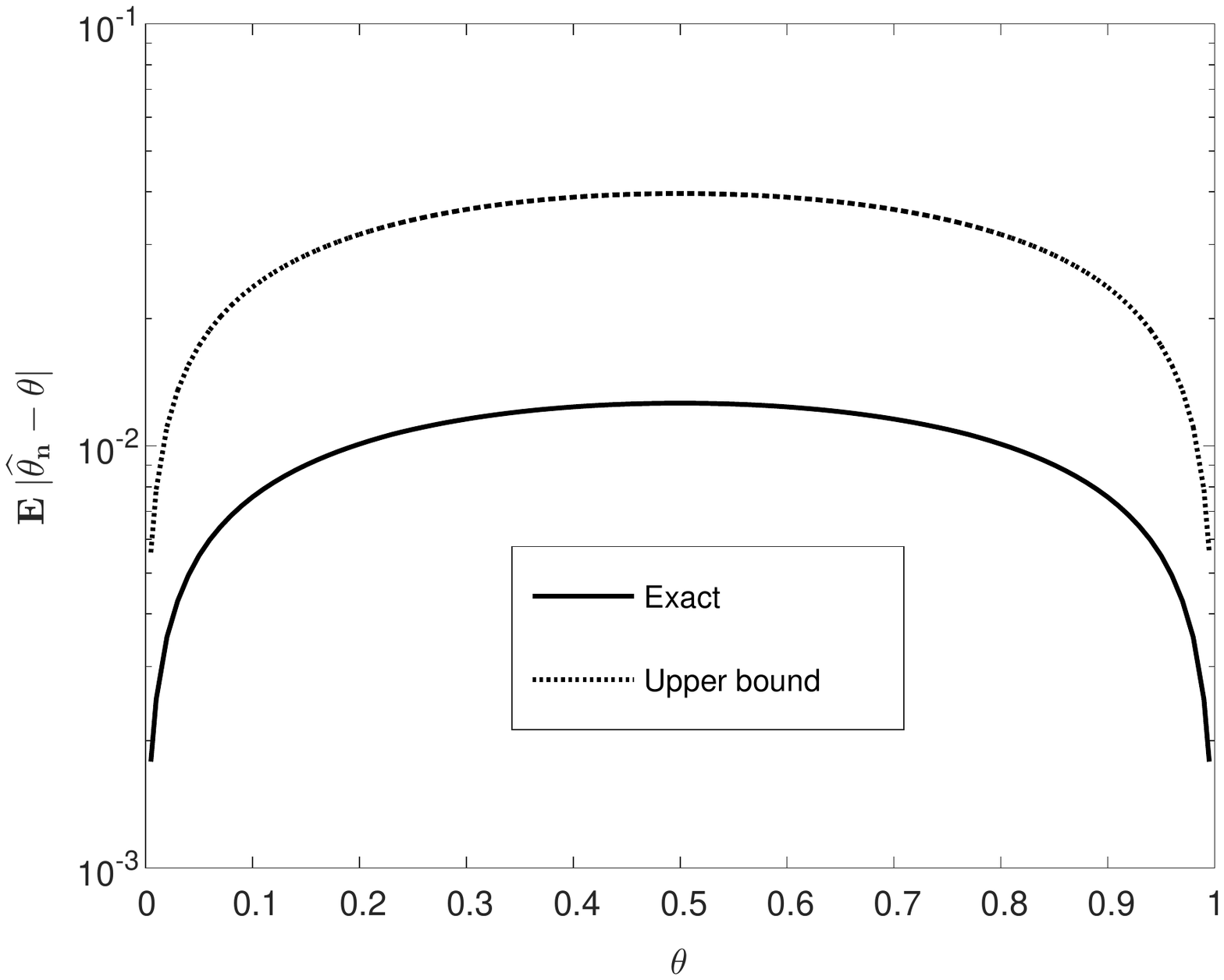}
\vspace*{-4cm}
\caption{\label{Fig1: estimation err}
$\bE \, \bigl| \widehat{\theta}_n - \theta \bigr|$
(see \eqref{20200111a1} and  \eqref{CF Bernoulli}) versus
its upper bound in \eqref{UB - estimation}
as functions of $\theta \in [0,1]$ with $n=1000$.}
\end{figure}
\begin{figure}[hbt]
\centering
\vspace*{-4.7cm}
\includegraphics[width=11cm]{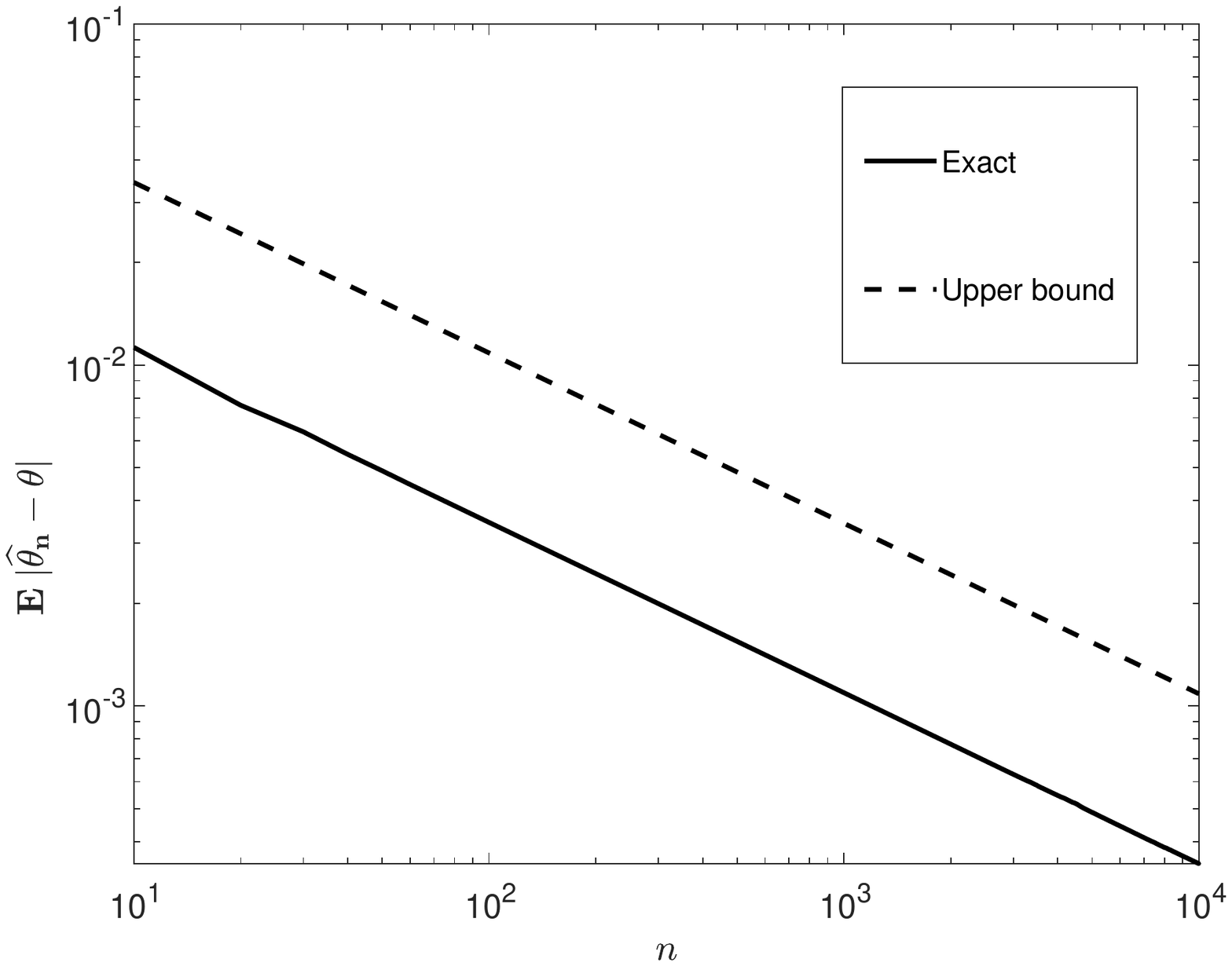}
\vspace*{-4cm}
\caption{\label{Fig2: estimation err} A plot of
$\bE \bigl| \widehat{\theta}_n - \theta \bigr|$
(see \eqref{20200111a1} and  \eqref{CF Bernoulli})
versus its upper bound in \eqref{UB - estimation}
as functions of $n$ with $\theta = \tfrac14$.}
\end{figure}

To conclude, this subsection provides an exact, double--integral expression
for the $\rho$-th moment of the estimation error of the expectation of $n$
i.i.d. random variables. In other words, the dimension of the integral
does not increase with $n$, and it is a calculable expression. We further
compare our expression with an upper bound that stems from concentration
inequalities. Although the scaling of the bound as a polynomial of $n$
is correct, the difference between the exact expression and the bound is
significant (see Fig.~\ref{Fig1: estimation err} and~\ref{Fig2: estimation err}).

\subsection{R\'{e}nyi Entropy of Extended Multivariate Cauchy Distributions}
\label{subsec: RE Cauchy}

Generalized Cauchy distributions, their mathematical properties, and applications
are of interest (see, e.g., \cite{Alzaatreh16, Carrillo10, MK18, MS_Ent20}). The Shannon
differential entropy of a family of generalized Cauchy distributions was derived in
\cite[Proposition~1]{Alzaatreh16}, and also, a lower bound on the differential entropy of
a family of extended multivariate Cauchy distributions (cf. \cite[Equation~(42)]{MK18}) was
derived in \cite[Theorem~6]{MK18}. Furthermore, an exact single-letter expression for the
differential entropy of the different family of extended multivariate Cauchy distributions
was recently derived in \cite[Section~3.1]{MS_Ent20}. Motivated by these studies, as
well as the various information-theoretic applications of R\'{e}nyi information measures, we
apply Theorem~\ref{thm: rho moment} to obtain the R\'{e}nyi (differential) entropy of an
arbitrary positive order $\alpha$ for the extended multivariate Cauchy distributions in
\cite[Section~3.1]{MS_Ent20}. As we shall see in this subsection, the integral representation
for the R\'{e}nyi entropy of the latter family of extended multivariate Cauchy distributions
is two-dimensional, irrespective of the dimension $n$ of the random vector.

Let $X^n = (X_1,\ldots,X_n)$ be a random vector whose probability density
function is of the form
\begin{align} \label{pdf gen. Cauchy}
f(x^n)=\frac{C_n}{\left[1+\sum_{i=1}^n g(x_i)\right]^q}, \quad
x^n = (x_1, \ldots, x_n) \in \reals^n,
\end{align}
for a certain function $g \colon \reals \to [0, \infty)$, and a positive
constant $q$ such that
\begin{align}
\int_{\reals^n} \frac{1}{\left[1+\sum_{i=1}^n g(x_i)\right]^q} \;
\mathrm{d}x^n < \infty.
\end{align}
We refer to this kind of density (see also \cite[Section~3.1]{MS_Ent20}) as
a {\it generalized multivariate Cauchy density} because the multivariate Cauchy
density function is the special case pertaining to the choices $g(x)=x^2$ and
$q=\tfrac12 (n+1)$.
The differential Shannon entropy of the generalized multivariate Cauchy
density was derived in \cite[Section~3.1]{MS_Ent20} using the integral
representation of the logarithm \eqref{ir}, where it was presented as a
two--dimensional integral.

We next extend the analysis of \cite{MS_Ent20} to differential R\'{e}nyi
entropies of an arbitrary positive order $\alpha$ (recall that the
differential R\'{e}nyi entropy is specialized to the differential Shannon
entropy at $\alpha=1$ \cite{Renyientropy}). We show that, for the
generalized multivariate Cauchy density, the differential R\'{e}nyi
entropy can be presented as a two--dimensional integral, rather than
an $n$--dimensional integral. Defining
\begin{align}
\label{eq:Z}
Z(t)\dfn\int_{-\infty}^{\infty}e^{-tg(x)} \, \mathrm{d}x, \quad t>0,
\end{align}
we get from \eqref{pdf gen. Cauchy} (see \cite[Section~3.1]{MS_Ent20}) that
\begin{align}
\label{C_n}
C_n=\frac{\Gamma(q)}{\displaystyle \int_0^\infty t^{q-1}e^{-t}Z^n(t) \, \mathrm{d}t}.
\end{align}
For $g(x)=|x|^\theta$, with a fixed $\theta > 0$,
\eqref{eq:Z} implies that
\begin{align} \label{eq:spec. Z}
Z(t)=\frac{2 \, \Gamma(1/\theta)}{\theta \, t^{1/\theta}}.
\end{align}
In particular, for $\theta=2$ and $q=\tfrac12 (n+1)$, we get
the multivariate Cauchy density from \eqref{pdf gen. Cauchy}.
In this case,
it follows from \eqref{eq:spec. Z} that $Z(t)=\sqrt{\frac{\pi}{t}}$
for $t>0$, and from \eqref{C_n}
\begin{align}
\label{cauchynormalization}
C_n=\frac{\Gamma\left(\frac{n+1}{2}\right)}{\pi^{(n+1)/2}}.
\end{align}

For $\alpha \in (0,1) \cup (1, \infty)$, the (differential) R\'{e}nyi
entropy of order $\alpha$ is given by
\begin{align}
h_{\alpha}(X^n) & \dfn \frac1{1-\alpha} \; \log \int_{\reals^n}
f^{\alpha}(x^n) \, \mathrm{d}x^n \nonumber \\[0.1cm]
\label{Renyi ent.}
& \, = \frac1{1-\alpha} \, \log \bE\bigl[ f^{\alpha-1}(X^n) \bigr].
\end{align}
Using the Laplace transform relation,
\begin{align} \label{Laplace}
\frac{1}{s^q}=\frac{1}{\Gamma(q)}\int_0^\infty
t^{q-1}e^{-st} \, \mathrm{d}t, \quad \forall \, q>0, \; \mathrm{Re}(s) > 0,
\end{align}
we obtain that, for $\alpha > 1$
(see Appendix~\ref{appendix: generalized multivariate Cauchy}),
\begin{align}
h_{\alpha}(X^n)
&= \frac{\alpha}{\alpha-1} \, \log \int_0^{\infty} t^{q-1} e^{-t} Z^n(t) \, \mathrm{d}t
+ \frac{\log \Gamma\bigl(q(\alpha-1)\bigr)}{\alpha-1} - \log \Gamma(q) \nonumber \\[0.1cm]
\label{20200125a1}
& \hspace*{0.4cm} - \frac1{\alpha-1} \, \log \int_0^{\infty} \int_0^{\infty} t^{q(\alpha-1)-1} u^{q-1} e^{-(t+u)}
\, Z^n(t+u) \, \mathrm{d}u \, \mathrm{d}t.
\end{align}
Otherwise, if $\alpha \in (0,1)$,
we distinguish between the following two cases:
\begin{enumerate}[1)]
\item If $\alpha = 1-\frac{m}{q}$ for some $m \in \{1, \ldots, q-1\}$, then
\begin{align}
h_{\alpha}(X^n) &= \frac{\alpha}{1-\alpha} \, \log C_n - \frac1{1-\alpha} \, \log \Gamma(q) \nonumber \\
\label{20200212a}
& \hspace*{0.5cm} + \frac1{1-\alpha} \, \log \Biggl( \sum_{\ell=0}^m \biggl\{ (-1)^{m-\ell} \int_0^\infty t^{q-1}
e^{-t} \varphi_n^{(\ell)}(t) \, \mathrm{d}t \biggr\} \Biggr),
\end{align}
with
\begin{align}
\label{20200212b}
\varphi_n(t) \dfn Z^n(t), \quad \forall \, t \geq 0.
\end{align}
\item Otherwise (i.e., if $\rho \dfn q(1-\alpha) \notin \naturals$), then
\begin{align}
& h_{\alpha}(X^n) \nonumber \\
&= -\log C_n + \frac1{1-\alpha} \, \log \Biggl(
\frac1{1+\rho} \, \sum_{\ell=0}^{\lfloor \rho \rfloor}
\frac{\beta_\ell(n)}{B(\ell+1,\rho+1-\ell)} \nonumber \\
\label{20200212c}
&\hspace*{2.8cm} + \frac{\rho \, \sin(\pi \rho) \, \Gamma(\rho)}{\pi}
\int_0^\infty \frac{e^{-u}}{u^{\rho+1}} \,
\Biggl( \, \sum_{j=0}^{\lfloor \rho \rfloor}
\biggl \{ \frac{(-1)^j \, \beta_j(n)}{j!} \; u^j \biggr\} \\
&\hspace*{7.4cm} - \frac{C_n}{\Gamma(q)}
\int_0^\infty t^{q-1} e^{-t} Z^n(t+u) \, \mathrm{d}t \Biggr) \Biggr), \nonumber
\end{align}
where $\beta_0 \dfn 1$, and for all $j \in \naturals$
\begin{align}
\label{20200212d}
\hspace*{-0.3cm} \beta_j(n) \dfn \frac{C_n}{\Gamma(q)} \sum_{\ell = 0}^j \Biggl\{
\frac{(-1)^{j-\ell}}{B(\ell+1,j-\ell+1)} \sum_{k=0}^{\ell} \biggl\{ (-1)^{\ell-k}
\binom{\ell}{k} \int_0^\infty t^{q-1} e^{-t} \varphi_n^{(k)}(t) \, \mathrm{d}t \biggr\} \Biggr\}.
\end{align}
\end{enumerate}
The proof of the integral expressions of the R\'{e}nyi entropy of order
$\alpha \in (0,1)$, as given in \eqref{20200125a1}--\eqref{20200212d}, is
provided in Appendix~\ref{appendix: generalized multivariate Cauchy}.

Once again, the advantage of these expressions, which do not seem to be very simple (at
least on the face of it), is that they only involve one-- or two--dimensional
integrals, rather than an expression of an $n$--dimensional integral (as it
could have been in the case of an $n$--dimensional density).

\subsection{Mutual Information Calculations for Communication Channels with Jamming}
\label{subsec:jamming}

Consider a channel that is fed by an input vector
$X^n = (X_1, \ldots, X_n) \in \cX^n$ and generates
an output vector $Y^n = (Y_1, \ldots, Y_n) \in \cY^n$,
where $\cX$ and $\cY$ are either finite, countably infinite
or continuous alphabets, and $\cX^n$ and $\cY^n$
are their $n$-th order Cartesian powers. Let the
conditional probability distribution of the channel be given by
\begin{align} \label{transition law - jamming}
p_{Y^n|X^n}(y^n | x^n)
= \frac1n \sum_{i=1}^n \Biggl\{ \prod_{j \neq i} q_{Y|X}(y_j | x_j)
\; r_{Y|X}(y_i|x_i) \Biggr\},
\end{align}
where $r_{Y|X}(\cdot | \cdot)$ and $q_{Y|X}(\cdot|\cdot)$
are given conditional probability distributions of $Y$ given $X$,
$x^n = (x_1, \ldots, x_n) \in \cX^n$ and
$y^n = (y_1, \ldots, y_n) \in \cY^n$.
This channel model refers to a discrete memoryless channel
(DMC), which is nominally given by
\begin{align}
\label{memoryless stationary channel}
q_{Y^n|X^n}(y^n | x^n) = \prod_{i=1}^n q_{Y|X}(y_i | x_i),
\end{align}
where one of the transmitted symbols is jammed at a uniformly distributed
random time, $i$, and the transition distribution of the jammed symbol
is given by $r_{Y|X}(y_i|x_i)$ instead of $q_{Y|X}(y_i|x_i)$.
The restriction to a single jammed symbol is made merely
for the sake of simplicity, but it can easily be extended.

We wish to evaluate how the jamming affects the mutual information
$I(X^n; Y^n)$. Clearly, when one talks about jamming, the mutual information is decreased,
but this is not part of the mathematical model, where the relation between
$r$ and $q$ has not been specified. Let the input distribution be given by
the product form
\begin{align}
\label{input distribution}
p_{X^n}(x^n) = \prod_{i=1}^n p_X(x_i), \quad x^n \in \cX^n.
\end{align}
The mutual information (in nats) is given by
\begin{align}
& I(X^n; Y^n) \nonumber \\
\label{MI}
& = h(Y^n) - h(Y^n | X^n) \\
\label{sum integrals}
& = \int_{\cX^n \times \cY^n} p_{X^n, Y^n}(x^n, y^n)
\, \ln p_{Y^n|X^n}(y^n | x^n) \, \mathrm{d}x^n \, \mathrm{d}y^n
- \int_{\cY^n} p_{Y^n}(y^n) \,
\ln p_{Y^n}(y^n) \, \mathrm{d}y^n.
\end{align}
For simplicity of notation, we henceforth omit the domains of integration
whenever they are clear from the context. We have,
\begin{align}
& \int p_{X^n, Y^n}(x^n, y^n)
\, \ln p_{Y^n|X^n}(y^n | x^n) \, \mathrm{d}x^n \, \mathrm{d}y^n \nonumber \\
&= \int p_{X^n, Y^n}(x^n, y^n)
\, \ln \biggl( \frac{p_{Y^n|X^n}(y^n | x^n)}{q_{Y^n|X^n}(y^n | x^n)} \biggr)
\; \mathrm{d}x^n \, \mathrm{d}y^n \nonumber \\
\label{sum integrals 2}
&\hspace*{0.3cm} + \int p_{X^n, Y^n}(x^n, y^n)
\, \ln q_{Y^n|X^n}(y^n | x^n) \, \mathrm{d}x^n \, \mathrm{d}y^n.
\end{align}
By using the logarithmic expectation in \eqref{log. expectation},
and the following equality (see \eqref{transition law - jamming}
and \eqref{memoryless stationary channel}):
\begin{align}
\label{p_n}
\frac{p_{Y^n|X^n}(y^n | x^n)}{q_{Y^n|X^n}(y^n | x^n)} = \frac1n
\sum_{i=1}^n \frac{r_{Y|X}(y_i | x_i)}{q_{Y|X}(y_i | x_i)},
\end{align}
we obtain (see Appendix~\ref{appendix: jamming}.1)
\begin{align}
& \int p_{X^n, Y^n}(x^n, y^n)
\, \ln \biggl( \frac{p_{Y^n|X^n}(y^n | x^n)}{q_{Y^n|X^n}(y^n | x^n)} \biggr)
\; \mathrm{d}x^n \, \mathrm{d}y^n \nonumber \\
\label{J1}
&= \int_0^\infty \frac{1}{u} \, \Bigl[ e^{-u} - f^{n-1}\Bigl(\frac{u}{n}\Bigr) \,
g\Bigl(\frac{u}{n}\Bigr) \Bigr] \, \mathrm{d}u,
\end{align}
where, for $u \geq 0$,
\begin{align}
\label{f}
f(u) := \int p_X(x) \, q_{Y|X}(y|x) \, \exp\biggl(-\frac{u \, r_{Y|X}(y|x)}{q_{Y|X}(y|x)}\biggr) \, \mathrm{d}x \, \mathrm{d}y, \\
\label{g}
g(u) := \int p_X(x) \, r_{Y|X}(y|x) \, \exp\biggl(-\frac{u \, r_{Y|X}(y|x)}{q_{Y|X}(y|x)}\biggr) \, \mathrm{d}x \, \mathrm{d}y.
\end{align}
Moreover, owing to the product form of $q_n$,
it is shown in Appendix~\ref{appendix: jamming}.2 that
\begin{align}
& \int p_{X^n, Y^n}(x^n, y^n) \, \ln q_{Y^n|X^n}(y^n | x^n) \, \mathrm{d}x^n \, \mathrm{d}y^n \nonumber \\
\label{J2}
&= \int p_X(x) \, r_{Y|X}(y|x) \, \ln q_{Y|X}(y|x) \, \mathrm{d}x \, \mathrm{d}y  \nonumber \\
& \hspace*{0.3cm} + (n-1) \int p_X(x) \, q_{Y|X}(y|x) \, \ln q_{Y|X}(y|x) \, \mathrm{d}x \, \mathrm{d}y.
\end{align}
Combining \eqref{sum integrals 2}, \eqref{J1} and \eqref{J2},  we
express $h(Y^n | X^n)$ as a double integral over $\cX \times \cY$,
independently of $n$ (rather than an integration over $\cX^n \times \cY^n$):
\begin{align}
h(Y^n | X^n)
&= \int_0^\infty \frac{1}{u} \, \Bigl[ f^{n-1}\Bigl(\frac{u}{n}\Bigr) \,
g\Bigl(\frac{u}{n}\Bigr) - e^{-u} \Bigr] \, \mathrm{d}u \nonumber \\
& \hspace*{0.4cm} - \int p_X(x) \, r_{Y|X}(y|x) \, \ln q_{Y|X}(y|x) \, \mathrm{d}x \, \mathrm{d}y \nonumber \\
\label{20200108a3}
& \hspace*{0.4cm} -(n-1) \int p_X(x) \, q_{Y|X}(y|x) \, \ln q_{Y|X}(y|x) \, \mathrm{d}x \, \mathrm{d}y.
\end{align}

We next calculate the differential channel output entropy,
$h(Y^n)$, induced by $p_{Y^n|X^n}(\cdot|\cdot)$.
From Appendix~\ref{appendix: jamming}.3,
\begin{align}
\label{output dist.}
p_{Y^n}(y^n) = \prod_{j=1}^n v(y_j) \cdot \frac1n \sum_{i=1}^n \frac{w(y_i)}{v(y_i)},
\end{align}
where, for all $y \in \cY$,
\begin{align}
\label{v}
v(y) := \int q_{Y|X}(y|x) \, p_X(x) \, \mathrm{d}x, \\
\label{w}
w(y) := \int r_{Y|X}(y|x) \, p_X(x) \, \mathrm{d}x.
\end{align}
By \eqref{ir}, the following identity holds for every positive random variable $Z$
(see Appendix~\ref{appendix: jamming}.3):
\begin{align}
\label{eq: mean of Z ln Z}
\bE\{Z \ln Z\} = \int_0^{\infty} \frac1{u} \, \Bigl[ M'_Z(0) \, e^{-u} - M'_Z(-u) \Bigr] \, \mathrm{d}u
\end{align}
where $M_Z(u) := \bE\{e^{uZ}\}$.
By setting $Z := \frac1n \sum_{i=1}^n \frac{w(V_i)}{v(V_i)}$
where $\{V_i\}_{i=1}^n$ are i.i.d. random variables with the density function $v$,
some algebraic manipulations give (see Appendix~\ref{appendix: jamming}.3)
\begin{align}
h(Y^n) &= \int_0^\infty \frac1{u} \Bigl[ t^{n-1}\Bigl(\frac{u}{n}\Bigr) \, s\Bigl(\frac{u}{n}\Bigr) - e^{-u} \Bigr]
\, \mathrm{d}u \nonumber \\
\label{entropy of Y vec}
& \hspace*{0.5cm} - \int w(y) \, \ln v(y) \, \mathrm{d}y - (n-1) \int v(y) \ln v(y) \, \mathrm{d}y,
\end{align}
where
\begin{align}
\label{s}
s(u) := \int w(y) \, \exp\biggl(-\frac{u \, w(y)}{v(y)}\biggr) \, \mathrm{d}y, \quad u \geq 0, \\
\label{t}
t(u) := \int v(y) \, \exp\biggl(-\frac{u \, w(y)}{v(y)}\biggr) \, \mathrm{d}y, \quad u \geq 0.
\end{align}
Combining \eqref{MI}, \eqref{20200108a3} and \eqref{entropy of Y vec},
we obtain the mutual information for the channel with jamming, which is given by
\begin{align}
I_p(X^n; Y^n) &= \int_0^\infty \frac1{u} \Bigl[ t^{n-1}\Bigl(\frac{u}{n}\Bigr) \, s\Bigl(\frac{u}{n}\Bigr)
- f^{n-1}\Bigl(\frac{u}{n}\Bigr) \, g\Bigl(\frac{u}{n}\Bigr) \Bigr] \, \mathrm{d}u \nonumber \\
& \hspace*{0.5cm} + \int p_X(x) \, r_{Y|X}(y|x) \, \ln q_{Y|X}(y|x) \, \mathrm{d}x \, \mathrm{d}y - \int w(y) \, \ln v(y) \, \mathrm{d}y \nonumber \\
& \hspace*{0.5cm} + (n-1) \biggl[ \int p_X(x) \, q_{Y|X}(y|x) \, \ln q_{Y|X}(y|x) \, \mathrm{d}x \, \mathrm{d}y - \int v(y) \ln v(y) \, \mathrm{d}y \biggr].
\end{align}

We next exemplify our results in the case where $q$ is a binary symmetric channel (BSC)
with crossover probability $\delta \in (0, \tfrac12)$, and $p$ is a BSC with a larger
crossover probability, $\varepsilon \in (\delta, \tfrac12]$.
We assume that the input bits are i.i.d.\ and equiprobable.
The specialization of our analysis to this setup is provided in Appendix~\ref{appendix: jamming}.4,
showing that the mutual information of the channel $p_{X^n, Y^n}$, fed by the binary
symmetric source, is given by
\begin{align}
\label{eq: MI BSC and jamming}
I_{p}(X^n; Y^n) &= n \ln 2 - d(\varepsilon \| \delta)
- H_{\mathrm{b}}(\varepsilon) - (n-1) H_{\mathrm{b}}(\delta) \\
& \hspace*{0.4cm} + \int_0^\infty \biggl\{ e^{-u} -
\biggl[ (1-\delta) \, \exp\biggl(-\frac{(1-\varepsilon)u}{(1-\delta)n}\biggr)
+ \delta \, \exp\biggl(-\frac{\varepsilon u}{\delta n}\biggr) \biggr]^{n-1} \nonumber \\[0.1cm]
& \hspace*{2.7cm} \cdot \biggl[ (1-\varepsilon) \, \exp\biggl(-\frac{(1-\varepsilon)u}{(1-\delta)n}\biggr)
+ \varepsilon \, \exp\biggl(-\frac{\varepsilon u}{\delta n}\biggr) \biggr] \biggr\}
\, \frac{\mathrm{d}u}{u},  \nonumber
\end{align}
where $H_{\mathrm{b}} \colon [0,1] \to [0, \ln 2]$ is the binary entropy function
\begin{align}
H_{\mathrm{b}}(x) := -x \ln(x) - (1-x) \ln(1-x), \quad x \in [0,1]
\end{align}
with the convention that $0 \ln 0 = 0$, and
\begin{align}
d(\varepsilon \| \delta) \dfn \varepsilon \, \ln \biggl( \frac{\varepsilon}{\delta} \biggr)
+ (1-\varepsilon) \, \ln \biggl( \frac{1-\varepsilon}{1-\delta} \biggr), \quad (\delta, \varepsilon) \in [0,1]^2
\end{align}
denotes the binary relative entropy.
By the data processing inequality, the mutual
information in \eqref{eq: MI BSC and jamming}
is smaller than that of the BSC with crossover probability~$\delta$:
\begin{align}
\label{eq: MI BSC and no jamming}
I_q(X^n; Y^n) &= n \bigl(\ln 2 - H_{\mathrm{b}}(\delta) \bigr).
\end{align}
Fig.~\ref{jamming1} refers to the case where $\delta = 10^{-3}$  and $n=128$.
Here $I_{q}(X^n; Y^n) = 87.71~\mathrm{nats}$, and
$I_p(X^n; Y^n)$ is decreased by 2.88~nats
due to the jammer (see Fig.~\ref{jamming1}).
\begin{figure}[hbt]
\centering
\vspace*{-4.7cm}
\includegraphics[width=11cm]{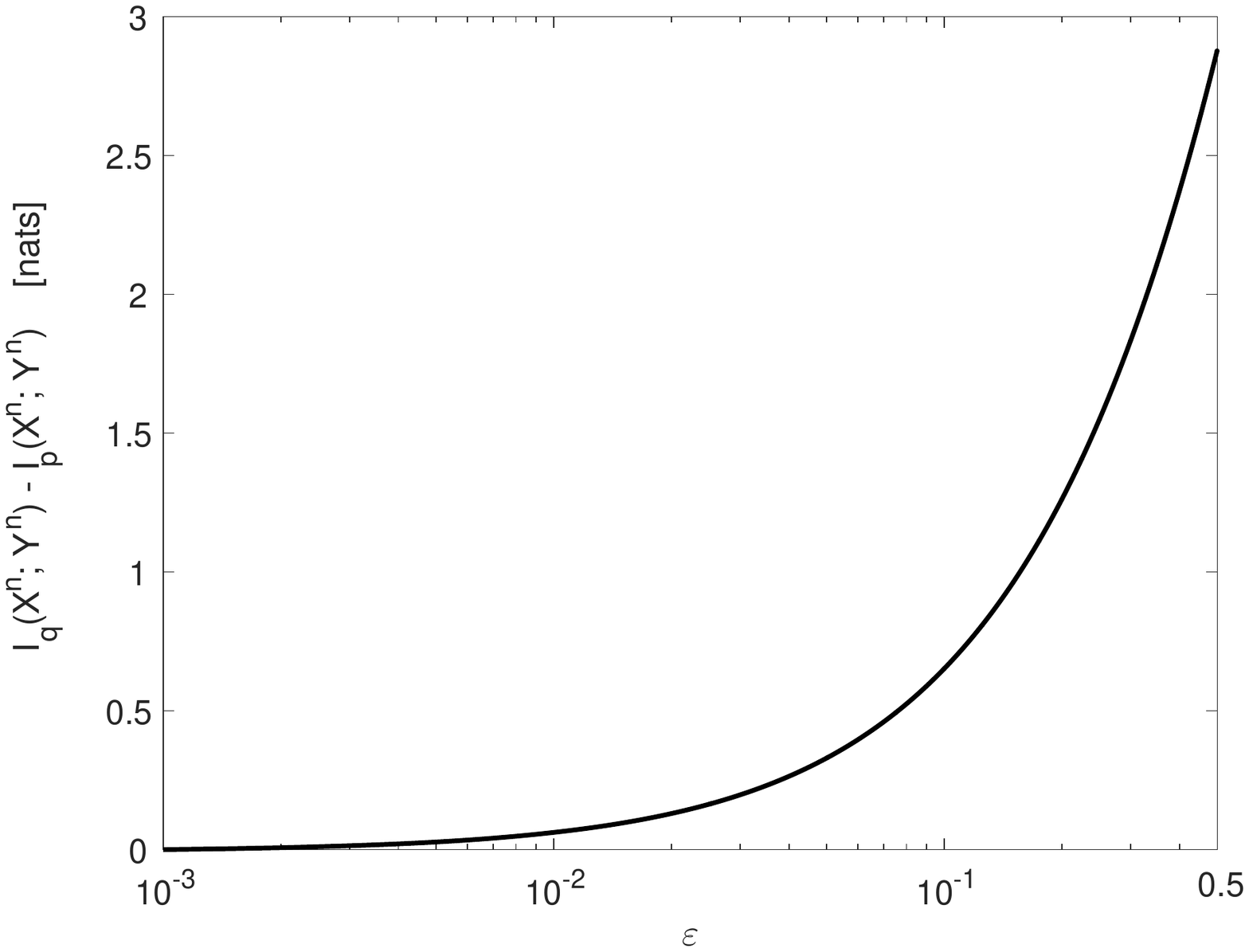}
\vspace*{-4cm}
\caption{The degradation in mutual information for $n=128$. The
jammer--free channel $q_{Y|X}$ is a BSC with crossover probability
$\delta = 10^{-3}$, and $r_{Y|X}$ for the jammed symbol is a BSC
with crossover probability $\varepsilon \in \bigl(\delta, \tfrac12 \bigr]$.
The input bits are i.i.d. and equiprobable. The degradation in
$I(X^n;Y^n)$ (nats) is displayed as a function of $\varepsilon$.}
\label{jamming1}
\end{figure}

Fig.~\ref{jamming2} refers to the case where $\delta = 10^{-3}$  and
$\varepsilon = \tfrac12$ (referring to complete
jamming of a single symbol which is chosen uniformly at random),
and it shows the difference in the mutual information $I(X^n; Y^n)$,
as a function of the length $n$, between the jamming-free BSC with
crossover probability $\delta$, and the channel with jamming.
\begin{figure}[hbt]
\centering
\vspace*{-4.7cm}
\includegraphics[width=11cm]{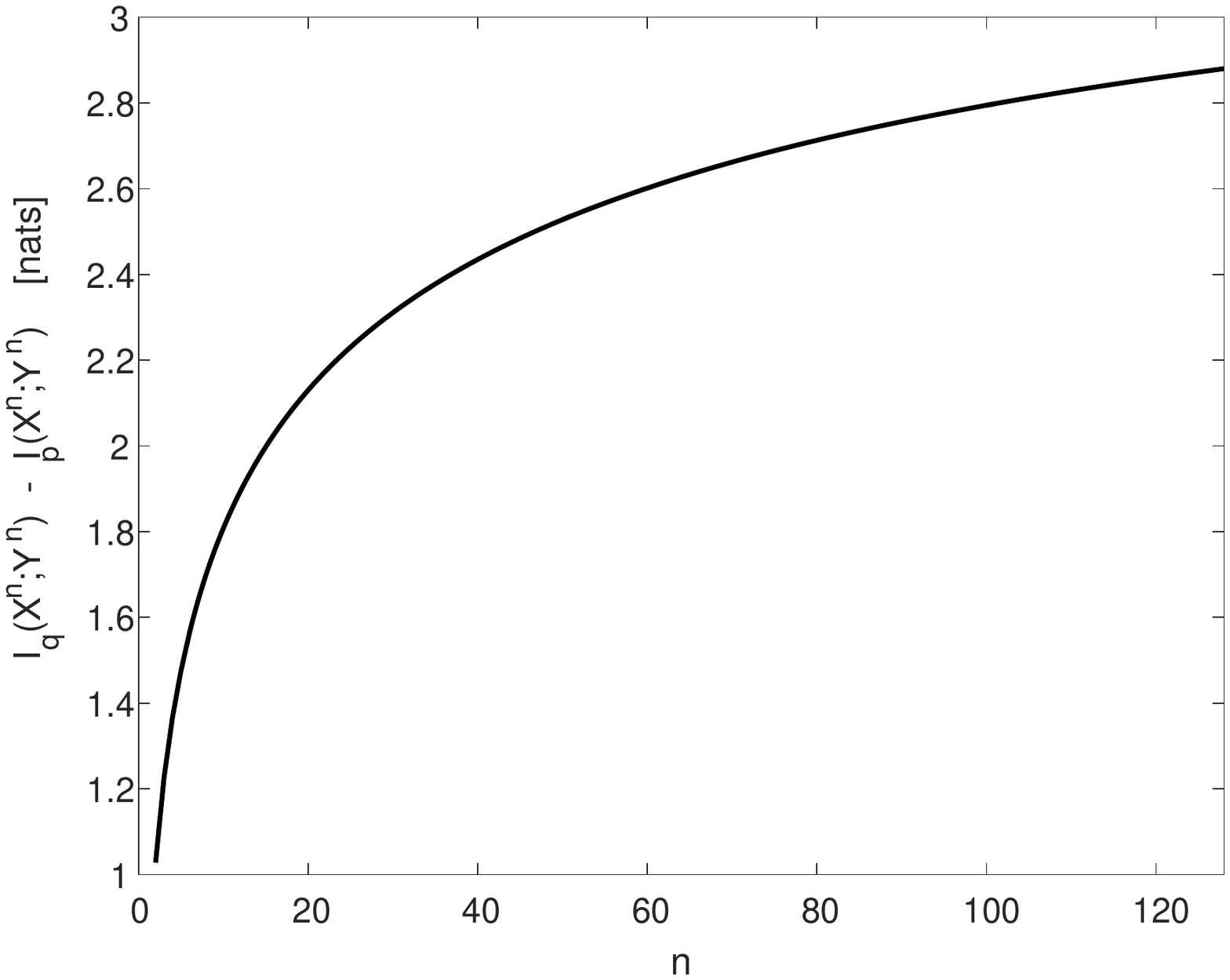}
\vspace*{-4cm}
\caption{The degradation in mutual information as a function of $n$.
The jammer--free channel $q_{Y|X}$ is a BSC with crossover probability
$\delta = 10^{-3}$, and $r_{Y|X}$ for the jammed symbol is a BSC with
crossover probability $\varepsilon = \tfrac12$.
The input bits are i.i.d. and equiprobable.}
\label{jamming2}
\end{figure}

To conclude, this subsection studies the change in the mutual information
$I(X^n; Y^n)$ due to jamming, relative to the mutual information associated
with the nominal channel without jamming. Due to the integral representations
provided in our analysis, the calculation of the mutual information finally
depends on one--dimensional integrals, as opposed to the original $n$-dimensional
integrals, pertaining to the expressions that define the associated differential
entropies.

\appendices
\renewcommand{\thesubsection}{\Alph{section}.\arabic{subsection}}

\section{Proof of Theorem~\ref{thm: rho moment}}
\label{appendix: theorem - rho moment}

\setcounter{equation}{0}
\renewcommand{\theequation}{\thesection.\arabic{equation}}

Let $\rho>0$ be a non-integer real, and define the function
$F_\rho \colon (0, \infty) \to \reals$  as follows:
\begin{align}
\label{eq: auxiliary function}
F_\rho(\mu) \dfn \int_0^\infty \frac1{u^{\rho+1}} \, \Biggl( e^{-\mu u}
- \sum_{j=0}^{\lfloor \rho \rfloor} \biggl\{ \frac{(-1)^j}{j!} \, (\mu-1)^j
u^j \biggr\} \, e^{-u} \Biggr) \, \mathrm{d}u, \quad \mu > 0,
\end{align}
with the convention that $0^0 := \underset{x \to 0^+}{\lim} x^x = 1$.
By the Taylor series expansion of $e^{-\mu u}$ as a function of $\mu$
around $\mu=1$, we find that for small positive $u$, the integrand
of \eqref{eq: auxiliary function} scales like
$u^{-(\rho - \lfloor \rho \rfloor)}$ with
$\rho - \lfloor \rho \rfloor \in (0,1)$. Furthermore, for large
$u$, the same integrand scales like $u^{-(\rho+1)} e^{-\min \{ \mu, 1 \} u}$.
This guarantees the convergence of the integral, and so $F_\rho(\cdot)$ is
well-defined and finite in the interval $(0, \infty)$.

From \eqref{eq: auxiliary function}, $F_{\rho}(1) = 0$ (for $\mu=1$, the
integrand of \eqref{eq: auxiliary function} is identically
zero on $(0, \infty)$). Differentiation $\ell$ times with respect to
$\mu$, under the integration sign with
$\ell \in \bigl\{0, \ldots, \lfloor \rho \rfloor \bigr\}$, gives
\begin{align}
\label{eq: first derivatives}
F_\rho^{(\ell)}(\mu) = \int_0^\infty \frac1{u^{\rho+1}}
\Biggl[ (-1)^\ell u^\ell e^{-\mu u} - \sum_{j=\ell}^{\lfloor \rho \rfloor}
\biggl\{ \frac{(-1)^j}{(j-\ell)!} \cdot (\mu-1)^{j-\ell} u^j \biggr\} \, e^{-u}
\Biggr] \, \mathrm{d}u,
\end{align}
which implies that
\begin{align}
\label{eq: first derivatives at 1}
F_{\rho}^{(\ell)}(1) = 0, \qquad \ell = 0, \ldots, \lfloor \rho \rfloor.
\end{align}
We next calculate $F_{\rho}^{(k)}(\mu)$ for $k \dfn \lfloor \rho \rfloor + 1$
and $\mu>0$:
\begin{align}
F_{\rho}^{(k)}(\mu) &= \int_0^\infty \frac1{u^{\rho+1}} \,
\frac{\partial^k}{\partial \mu^k} \Biggl\{ e^{-\mu u}
- \sum_{j=0}^{\lfloor \rho \rfloor} \biggl\{ \frac{(-1)^j}{j!} \cdot
(\mu-1)^j u^j \biggr\} \, e^{-u} \Biggr\} \, \mathrm{d}u \nonumber \\
&= \int_0^\infty \frac{(-u)^k e^{-\mu u}}{u^{\rho+1}} \; \mathrm{d}u \nonumber \\
&= (-1)^k \int_0^\infty u^{k-\rho-1} e^{-\mu u} \, \mathrm{d}u \nonumber \\
&= (-1)^k \int_0^\infty \Bigl(\frac{t}{\mu}\Bigr)^{k-\rho-1} \, e^{-t} \, \mu^{-1}
\, \mathrm{d}t \nonumber \\
\label{eq: k-th derivative}
&= (-1)^k \mu^{\rho-k} \, \Gamma(k-\rho).
\end{align}
Hence, from \eqref{eq: first derivatives at 1} and \eqref{eq: k-th derivative},
\begin{align}
\label{0120a}
& F_{\rho}(1) = \ldots = F_{\rho}^{(\lfloor \rho \rfloor)}(1) = 0, \\
\label{0120b}
& F_{\rho}^{(k)}(\mu) = (-1)^k \mu^{\rho-k} \, \Gamma(k-\rho),
\qquad k \dfn \lfloor \rho \rfloor + 1, \; \mu > 0.
\end{align}
By integrating both sides of \eqref{0120b} with respect to $\mu$,
successively $k$ times, \eqref{0120a} implies that
\begin{align}
\label{0120c}
F_\rho(\mu) = \frac{(-1)^k \; \Gamma(k-\rho) \; \mu^\rho}{ \underset{i=0}{\overset{k-1}{\prod}} (\rho-i)}
+ \sum_{i=0}^{k-1} c_i(\rho) \, (\mu-1)^i, \qquad k \dfn \lfloor \rho \rfloor + 1, \; \mu>0,
\end{align}
with some integration constants $\bigl\{c_i(\rho)\bigr\}_{i=0}^{k-1}$.
Since $F_{\rho}(1)=0$ (see \eqref{0120a}), \eqref{0120c} implies that
\begin{align}
\label{eq: c0}
c_0(\rho) = \frac{(-1)^{k+1} \; \Gamma(k-\rho)}{ \underset{i=0}{\overset{k-1}{\prod}} (\rho-i)},
\end{align}
and since (by assumption) $\rho$ is a non-integer, the denominator
on the right--hand side of \eqref{eq: c0} is non-zero.
Moreover, since $F_{\rho}^{(\ell)}(1)=0$ for all $\ell \in \{1, \ldots k-1\}$
(see \eqref{0120a}), differentiation of both
sides of \eqref{0120c} $\ell$ times at $\mu=1$ yields
\begin{align}
\label{eq: c_ell}
c_{\ell}(\rho) := \frac{(-1)^{k+1} \, \Gamma(k-\rho) \; \underset{i=0}{\overset{\ell-1}{\prod}}
(\rho-i)}{\ell! \, \underset{i=0}{\overset{k-1}{\prod}} (\rho-i)},
\qquad \ell = 1, \ldots, k-1.
\end{align}
Substituting \eqref{eq: c0} and \eqref{eq: c_ell} into \eqref{0120c} gives
\begin{align}
\label{eq: F as sum}
F_\rho(\mu) = \frac{(-1)^k \; \Gamma(k-\rho)}{ \underset{i=0}{\overset{k-1}{\prod}} (\rho-i)}
\left[ \mu^\rho - 1 - \sum_{\ell=1}^{k-1} \Biggl\{ \, \frac1{\ell!} \,
\prod_{i=0}^{\ell-1} (\rho-i) \; (\mu-1)^\ell \, \Biggr\} \right], \quad \mu > 0.
\end{align}
Combining \eqref{eq: auxiliary function} with \eqref{eq: F as sum} and
rearranging terms, we obtain
\begin{align}
\mu^\rho  &= 1 + \sum_{\ell=1}^{k-1} \Biggl\{ \frac1{\ell!} \,
\prod_{i=0}^{\ell-1} (\rho-i) \; (\mu-1)^\ell \, \Biggr\} \nonumber \\
\label{0120d}
& \hspace*{0.4cm} + \frac{(-1)^{k-1} \, \underset{i=0}{\overset{k-1}{\prod}}
(\rho-i)}{\Gamma(k-\rho)} \int_0^\infty \frac1{u^{\rho+1}} \, \Biggl( \,
\sum_{j=0}^{\lfloor \rho \rfloor} \biggl\{ \frac{(-1)^j}{j!} \,
(\mu-1)^j u^j \biggr\} \, e^{-u} - e^{-\mu u} \Biggr) \, \mathrm{d}u.
\end{align}
Setting $\mu := X \geq 0$, and taking expectations of both sides of \eqref{0120d}
yield (see \eqref{eq: MGF} and \eqref{eq: alpha_k})
\begin{align}
\bE \bigl\{X^\rho\bigr\}
&= 1 + \sum_{\ell=1}^{k-1} \Biggl\{ \, \frac1{\ell!} \,
\prod_{i=0}^{\ell-1} (\rho-i) \; \alpha_\ell \, \Biggr\} \nonumber \\
\label{0120e}
& \hspace*{0.4cm} + \frac{(-1)^{k-1} \, \underset{i=0}{\overset{k-1}{\prod}}
(\rho-i)}{\Gamma(k-\rho)} \int_0^\infty \frac1{u^{\rho+1}} \, \Biggl( \,
\sum_{j=0}^{\lfloor \rho \rfloor} \biggl\{ \frac{(-1)^j \, \alpha_j}{j!} \;
u^j \biggr\} \, e^{-u} - M_X(-u) \Biggr) \, \mathrm{d}u.
\end{align}
We next rewrite and simplify both terms in the right side of \eqref{0120e} as follows:
\begin{align}
1 + \sum_{\ell=1}^{k-1} \Biggl\{ \, \frac1{\ell!} \, \prod_{i=0}^{\ell-1} (\rho-i)
\; \alpha_\ell \, \Biggr\}
\label{0120f}
&= 1 + \sum_{\ell=1}^{k-1} \Biggl\{ \, \frac1{\Gamma(\ell+1)} \;
\frac{\Gamma(\rho+1)}{\Gamma(\rho-\ell+1)} \cdot \alpha_\ell \, \Biggr\} \\[0.05cm]
\label{0120f.2}
&= 1 + \frac1{1+\rho} \; \sum_{\ell=1}^{k-1} \Biggl\{ \, \frac1{\Gamma(\ell+1)} \;
\frac{\Gamma(\rho+2)}{\Gamma(\rho-\ell+1)} \cdot \alpha_\ell \, \Biggr\} \\[0.05cm]
\label{0120g}
&= 1 + \frac1{1+\rho} \; \sum_{\ell=1}^{k-1} \frac{\alpha_\ell}{B(\ell+1, \rho-\ell+1)} \\[0.05cm]
\label{0120g2}
&= \frac1{1+\rho} \; \sum_{\ell=0}^{k-1} \frac{\alpha_\ell}{B(\ell+1, \rho-\ell+1)},
\end{align}
and
\begin{align}
\frac{(-1)^{k-1} \, \underset{i=0}{\overset{k-1}{\prod}} (\rho-i)}{\Gamma(k-\rho)}
\label{0120h}
&= \frac{(-1)^{k-1} \, \Gamma(\rho+1)}{\Gamma(k-\rho) \, \Gamma(\rho-k+1)} \\[0.05cm]
\label{0120i}
&= (-1)^{k-1} \, \Gamma(\rho+1) \cdot \frac{\sin \bigl(\pi (k-\rho) \bigr)}{\pi} \\[0.05cm]
\label{0120j}
&= \frac{\rho \, \sin(\pi \rho) \, \Gamma(\rho)}{\pi}.
\end{align}
Eqs.~\eqref{0120f}, \eqref{0120f.2}, \eqref{0120h} and \eqref{0120j} are based on the
recursion (see, e.g., \cite[page~904, Identity~(8.331)]{GR14})
\begin{align}
\label{eq: Gamma recursion}
\Gamma(x+1) = x \, \Gamma(x), \qquad x>0,
\end{align}
\eqref{0120g} relies on the relation between the Beta and Gamma functions in \eqref{eq1: Beta};
\eqref{0120g2} is based on the following equality (see \eqref{eq1: Beta}, \eqref{eq: Gamma recursion},
and recall that $\Gamma(1)=1$):
\begin{align}
B(1, \rho+1) = \frac{\Gamma(1) \, \Gamma(\rho+1)}{\Gamma(\rho+2)} = \frac1{\rho+1},
\end{align}
and, finally, \eqref{0120j} holds by using the identity (see, e.g., \cite[page~905, Identity~(8.334)]{GR14})
\begin{align}
\label{eq: Gamma identity 1}
& \Gamma(x) \, \Gamma(1-x) = \frac{\pi}{\sin(\pi x)}, \quad \forall \, x \in (0,1),
\end{align}
with $x: = k - \rho = \lfloor \rho \rfloor + 1 - \rho \in (0,1)$ (since,
by assumption, $\rho$ is a non-integer).
Combining \eqref{0120e}--\eqref{0120j} gives \eqref{eq: rho>0 non-int.} (recall
that $\alpha_0 \dfn 1$, and $ k - 1 \dfn \lfloor \rho \rfloor$ holds by \eqref{0120b}).

We finally prove \eqref{eq2: alpha_k}. By \eqref{eq: alpha_k}, for all
$j \in \naturals$,
\begin{align}
\alpha_j &= \bE \bigl\{(X-1)^j \bigr\} \nonumber \\
&= \sum_{\ell=0}^j (-1)^{j-\ell} \binom{j}{\ell} \, \bE\bigl\{ X^\ell \bigr\} \nonumber \\
&= \sum_{\ell=0}^j \frac{(-1)^{j-\ell} \; \Gamma(j+1) \; M_X^{(\ell)}(0)}{\Gamma(\ell+1) \, \Gamma(j-\ell+1)}
\nonumber \\
&= \frac1{j+1} \, \sum_{\ell=0}^j \frac{(-1)^{j-\ell} \; \Gamma(j+2) \; M_X^{(\ell)}(0)}{\Gamma(\ell+1)
\, \Gamma(j-\ell+1)} \nonumber \\
&= \frac1{j+1} \, \sum_{\ell=0}^j \frac{(-1)^{j-\ell} \, M_X^{(\ell)}(0)}{B(\ell+1, j-\ell+1)}.
\end{align}

\section{Complementary Details of the Analysis in Section~\ref{subsec: estimation}}
\label{appendix: estimation}

\setcounter{equation}{0}
\renewcommand{\theequation}{\thesection.\arabic{equation}}

\subsection*{{\em B.1} Proof of Eq.~\eqref{MGF D_n}}

For all $u>0$,
\begin{align}
M_{D_n}(-u)
\label{20200110a7.0}
&= \bE \Bigl\{ \exp\Bigl(-u \bigl( \widehat{\theta}_n - \theta \bigr)^2 \Bigr) \Bigr\} \\[0.1cm]
&= \bE \biggl\{ \frac1{2 \sqrt{\pi u}} \int_{-\infty}^\infty e^{j \omega (\widehat{\theta}_n - \theta)}
\label{20200110a7}
\, e^{-\omega^2/(4u)} \, \mathrm{d}\omega \biggr\} \\[0.1cm]
\label{20200110a8}
&= \frac1{2 \sqrt{\pi u}} \int_{-\infty}^\infty e^{-j \omega \theta} \, \bE\bigl\{e^{j \omega \widehat{\theta}_n}
\bigr\} \, e^{-\omega^2/(4u)} \, \mathrm{d}\omega \\[0.1cm]
\label{20200110a9}
&= \frac1{2 \sqrt{\pi u}} \int_{-\infty}^\infty e^{-j \omega \theta} \; \bE\biggl\{\exp\Bigl(\frac{j \omega}{n} \,
\sum_{i=1}^n X_i \Bigr) \biggr\} \, e^{-\omega^2/(4u)} \, \mathrm{d}\omega \\[0.1cm]
\label{20200110a10}
&= \frac1{2 \sqrt{\pi u}} \int_{-\infty}^\infty e^{-j \omega \theta} \, \phi_X^{n}\biggl(\frac{\omega}{n}\biggr) \,
e^{-\omega^2/(4u)} \, \mathrm{d}\omega,
\end{align}
where \eqref{20200110a7.0} is \eqref{20200110a3};
\eqref{20200110a7} relies on \eqref{eq: Gaussian identity}; \eqref{20200110a8}
holds by interchanging expectation and integration; \eqref{20200110a9} is due to
\eqref{estimator}, and \eqref{20200110a10} holds by the assumption
that $X_1, \ldots, X_n$ are i.i.d.

\subsection*{{\em B.2} Derivation of the Upper bound in \eqref{UB - estimation}}
\label{appendix: estimation 2}

For all $\rho > 0$,
\begin{align}
\bE \bigl\{ \bigl| \widehat{\theta}_n - \theta \bigr|^{\rho} \bigr\}
&= \int_0^\infty \prob \bigl( \bigl| \widehat{\theta}_n - \theta \bigr|^{\rho}
\geq t \bigr) \, \mathrm{d}t \nonumber \\
&= \int_0^\infty \prob \bigl( \bigl| \widehat{\theta}_n - \theta \bigr|^{\rho}
\geq \varepsilon^\rho \bigr) \, \rho \, \varepsilon^{\rho-1} \, \mathrm{d}\varepsilon \nonumber \\
\label{20200113a1}
&= \int_0^\infty \prob \bigl( \bigl| \widehat{\theta}_n - \theta \bigr|
\geq \varepsilon \bigr) \, \rho \, \varepsilon^{\rho-1} \, \mathrm{d}\varepsilon.
\end{align}
We next use the Chernoff bound for upper bounding $\prob \bigl( \bigl|
\widehat{\theta}_n - \theta \bigr| \geq \varepsilon \bigr)$ for all $\varepsilon > 0$,
\begin{align}
\prob \bigl( \widehat{\theta}_n - \theta \geq \varepsilon \bigr)
&=\prob \Biggl( \sum_{i=1}^n (X_i - \theta) \geq n \varepsilon \Biggr) \nonumber \\
&\leq \inf_{s \geq 0} \, \Biggl\{ e^{-sn \varepsilon} \, \bE \biggl\{
\exp\biggl(s \sum_{i=1}^n (X_i - \theta)\biggr) \biggr\} \Biggr\} \nonumber \\
&= \inf_{s \geq 0} \, \biggl\{ e^{-sn \varepsilon} \, \prod_{i=1}^n
\bE \Bigl\{ e^{s (X_i - \theta)} \Bigr\} \biggr\} \nonumber \\
&= \inf_{s \geq 0} \, \biggl\{ e^{-sn \varepsilon} \, \Bigl( \theta \, e^{s(1-\theta)}
+ (1-\theta) \, e^{-s \theta} \Bigr)^n \biggr\} \nonumber \\
\label{20200113a2}
&= \inf_{s \geq 0} \, \Bigl\{ e^{-n s \varepsilon + n H_{\theta}(s)} \Bigr\}
\end{align}
with $\theta \in [0,1]$, and
\begin{align}
\label{function H_theta}
H_{\theta}(s) \dfn \ln \Bigl( \theta \, e^{s(1-\theta)} + (1-\theta) \,
e^{-s \theta} \Bigr), \quad s \geq 0.
\end{align}
We now use an upper bound on $H_{\theta}(s)$ for every $s \geq 0$.
By Theorem~3.2 and Lemma~3.3 in \cite{BK_ECP13} (see
also \cite[Lemma~2.4.6]{MRIS_FnT19}), we have
\begin{align}
\label{20200113a3}
H_{\theta}(s) \leq C(\theta) \, s^2
\end{align}
with
\begin{align}
\label{20200113a4}
C(\theta) \dfn
\begin{dcases}
\hspace*{0.6cm} 0,  \quad & \mbox{if} \; \theta = 0, \\
\frac{1-2\theta}{4 \ln \Bigl( \frac{1-\theta}{\theta} \Bigr)},
\quad & \mbox{if} \; \theta \in \bigl(0, \tfrac12 \bigr), \\
\tfrac12 \, \theta (1-\theta), \quad & \mbox{if} \;
\theta \in \bigl[ \tfrac12, 1 \bigr].
\end{dcases}
\end{align}
Combining \eqref{20200113a2} and \eqref{20200113a3} yields
\begin{align}
\prob \bigl( \widehat{\theta}_n - \theta \geq \varepsilon \bigr)
& \leq \inf_{s \geq 0} \, \Bigl\{ e^{-n \varepsilon s + n C(\theta) s^2} \Bigr\} \nonumber \\
&= \exp \biggl( -\frac{n \varepsilon^2}{4 C(\theta)} \biggr).
\end{align}
Similarly, it is easy to show that the same Chernoff bound applies also to
$\prob \bigl( \widehat{\theta}_n - \theta \leq -\varepsilon \bigr)$,
which overall gives
\begin{align}
\label{20200114a1}
\prob \bigl( \bigl| \widehat{\theta}_n - \theta \bigr| \geq \varepsilon \bigr)
\leq 2 \exp \biggl( -\frac{n \varepsilon^2}{4 C(\theta)} \biggr).
\end{align}
Inequality~\eqref{20200114a1} is a refined version of Hoeffding's inequality
(see \cite[Section~2.4.4]{MRIS_FnT19}), which is derived for the Bernoulli
distribution (see \eqref{20200113a2})
and by invoking the Chernoff bound; moreover, \eqref{20200114a1} coincides
with Hoeffding's inequality in the special case $\theta= \tfrac12$
(which, from \eqref{20200113a4}, yields $C(\theta) = \tfrac18$).
In view of the fact that \eqref{20200114a1} forms a specialization of
\cite[Theorem~2.4.7]{MRIS_FnT19}, it follows that the Bernoulli case
is the worst one (in the sense of leading to the looser upper bound)
among all probability distributions whose support is the interval $[0,1]$
and whose expected value is $\theta \in [0,1]$. However, in the Bernoulli case,
a simple symmetry argument applies for improving the bound \eqref{20200114a1} as follows.
Since $\{X_i\}$  are i.i.d., Bernoulli with mean $\theta$, then obviously, $\{1-X_i\}$
are Bernoulli, i.i.d.\ with mean $1-\theta$ and (from \eqref{estimator})
\begin{align}
\label{20200114a2}
\widehat{\theta}_n(1-X_1, \ldots, 1-X_n) = 1 - \widehat{\theta}_n(X_1, \ldots, X_n),
\end{align}
which implies that the error estimation is identical in both cases. Hence,
$\prob \bigl( \bigl| \widehat{\theta}_n - \theta \bigr| \geq \varepsilon \bigr)$
is symmetric around $\theta = \tfrac12$. It can be verified that
\begin{align}
\label{20200114a3}
\min \bigl\{ C(\theta), \, C(1-\theta) \bigr\} = \tfrac12 \, \theta (1-\theta),
\quad \forall \, \theta \in [0,1],
\end{align}
which follows from \eqref{20200113a4} and since $C(\theta) > C(1-\theta)$
for all $\theta \in (0, \tfrac12)$ (see \cite[Fig.~2.1]{MRIS_FnT19}).
In view of \eqref{20200114a3} and the above symmetry consideration,
the upper bound in \eqref{20200114a1} is improved for values
of $\theta \in (0, \tfrac12)$, which therefore gives
\begin{align}
\label{20200114a4}
\prob \bigl( \bigl| \widehat{\theta}_n - \theta \bigr| \geq \varepsilon \bigr)
\leq 2 \exp \biggl( -\frac{n \varepsilon^2}{2\theta(1-\theta)} \biggr), \quad
\forall \, \theta \in [0,1], \; \varepsilon > 0.
\end{align}
From \eqref{estimator}, the
probability in \eqref{20200114a4} vanishes if $\theta=0$
or $\theta=1$. Consequently, for $\rho > 0$,
\begin{align}
\label{20200114a5}
\bE \bigl\{ \bigl| \widehat{\theta}_n - \theta \bigr|^{\rho} \bigr\}
&= \int_0^\infty \prob \bigl( \bigl| \widehat{\theta}_n - \theta \bigr|
\geq \varepsilon \bigr) \, \rho \, \varepsilon^{\rho-1} \, \mathrm{d}\varepsilon \\
\label{20200114a6}
&\leq \int_0^\infty 2 \exp \biggl( -\frac{n \varepsilon^2}{2\theta(1-\theta)} \biggr) \,
\rho \, \varepsilon^{\rho-1} \, \mathrm{d}\varepsilon \\
\label{20200114a7}
&= \rho \, \bigl(2 \theta(1-\theta) \bigr)^{\rho/2}
\int_0^\infty u^{\rho/2-1} \, e^{-u} \, \mathrm{d}u \cdot n^{-\rho/2} \\
\label{20200114a8}
&= \rho \, \Gamma\Bigl(\frac{\rho}{2}\Bigr) \bigl(2 \theta(1-\theta) \bigr)^{\rho/2}
\cdot n^{-\rho/2} \\
\label{20200114a9}
&= K(\rho, \theta) \cdot n^{-\rho/2},
\end{align}
where \eqref{20200114a5}--\eqref{20200114a9} hold, respectively, due to \eqref{20200113a1},
\eqref{20200114a4}, the substitution $u \dfn \frac{n \varepsilon^2}{2 \theta (1-\theta)}$,
\eqref{eq: Gamma} and \eqref{K - UB}.

\section{Complementary Details of the Analysis in Section~\ref{subsec: RE Cauchy}}
\label{appendix: generalized multivariate Cauchy}

\setcounter{equation}{0}
\renewcommand{\theequation}{\thesection.\arabic{equation}}

We start by proving \eqref{20200125a1}.
In view of  \eqref{Renyi ent.}, for $\alpha \in (0,1) \cup (1, \infty)$
\begin{align}
\label{20200210}
h_{\alpha}(X^n) = \frac1{1-\alpha} \, \log \bE\bigl[ f^{\alpha-1}(X^n) \bigr],
\end{align}
where $X^n \dfn (X_1, \ldots, X_n)$. For $\alpha > 1$, we get
\begin{align}
& \bE \bigl[f^{\alpha-1}(X^n)\bigr] \nonumber \\
\label{20200211a}
&= C_n^{\alpha-1} \; \bE \Biggl\{ \biggl[ 1 + \sum_{i=1}^n g(X_i) \biggr]^{q(1-\alpha)} \Biggr\} \\
\label{20200211b}
&= C_n^{\alpha-1} \; \int_{\reals^n} f(x^n) \cdot \frac1{\Gamma\bigl(q (\alpha-1) \bigr)}
\int_0^{\infty} t^{q(\alpha-1)-1}
\exp\Biggl\{-\Biggl(1 + \sum_{i=1}^n g(x_i) \Biggr) t \Biggr\} \, \mathrm{d}t \\[0.1cm]
\label{20200211c}
&= \frac{C_n^{\alpha-1}}{\Gamma\bigl(q (\alpha-1) \bigr)} \int_0^\infty t^{q(\alpha-1)-1} \, e^{-t} \;
\bE\Biggl[ \exp \Biggl(-t \sum_{i=1}^n g(X_i) \Biggr) \Biggr] \, \mathrm{d}t.
\end{align}
where \eqref{20200211a} holds due to \eqref{pdf gen. Cauchy}; \eqref{20200211b}
follows from \eqref{Laplace}, and \eqref{20200211c} holds by swapping order of
integrations. Furthermore, from \eqref{pdf gen. Cauchy} and \eqref{Laplace},
\begin{align}
f(x^n) &= \frac{C_n}{ \Bigl(1 + \sum_{i=1}^n g(x_i) \Bigr)^q } \nonumber\\
\label{20200211d}
&= \frac{C_n}{\Gamma(q)} \int_0^\infty u^{q-1} e^{-u}
\exp \biggl(-u \sum_{i=1}^n g(x_i) \biggr) \mathrm{d}u,
\quad \forall \, x^n \in \reals^n,
\end{align}
and it follows from \eqref{20200211d} and by swapping order of integrations,
\begin{align}
& \bE\Biggl[ \exp \biggl(-t \sum_{i=1}^n g(X_i) \biggr) \Biggr] \nonumber \\
&= \frac{C_n}{\Gamma(q)} \int_0^\infty u^{q-1} e^{-u} \int_{\reals^n}
\exp \Bigl( -(t+u) \sum_{i=1}^n g(x_i) \Bigr) \mathrm{d}x^n \, \mathrm{d}u \nonumber \\
&= \frac{C_n}{\Gamma(q)} \int_0^\infty u^{q-1} e^{-u} \Biggl\{ \prod_{i=1}^n \int_{-\infty}^{\infty}
\exp\bigl(-(t+u) g(x_i) \bigr) \, \mathrm{d}x_i  \Biggr\} \mathrm{d}u \nonumber \\
&= \frac{C_n}{\Gamma(q)} \int_0^\infty u^{q-1} e^{-u} \biggl( \int_{-\infty}^{\infty}
\exp\bigl(-(t+u) \, g(x) \bigr) \, \mathrm{d}x \biggr)^n \, \mathrm{d}u \nonumber \\
\label{20200211e}
&= \frac{C_n}{\Gamma(q)} \int_0^\infty u^{q-1} e^{-u} Z^n(t+u) \, \mathrm{d}u
\end{align}
where \eqref{20200211e} holds by the definition of $Z(\cdot)$ in \eqref{eq:Z}.
Finally, combining \eqref{C_n}, \eqref{20200210}, \eqref{20200211c} and
\eqref{20200211e} gives \eqref{20200125a1}.

The proof of \eqref{20200212a}--\eqref{20200212d} is a straightforward calculation
which follows by combining \eqref{20200210}, \eqref{20200211a}, \eqref{20200211e}
and Theorem~\ref{thm: rho moment} (we replace $\{\alpha_j\}$ in Theorem~\ref{thm: rho moment}
with $\{\beta_j(n)\}$ in order not to confuse with the order $\alpha$ of the R\'{e}nyi entropy
of $X^n$).

\section{Calculations of the $n$-Dimensional Integrals in Section~\ref{subsec:jamming}}
\label{appendix: jamming}

\setcounter{equation}{0}
\renewcommand{\theequation}{\thesection.\arabic{equation}}

\subsection*{{\em D.1} Proof of Eqs.~\eqref{J1}--\eqref{g}}

\vspace*{-0.8cm}
\begin{align}
& \int p_{X^n, Y^n}(x^n, y^n)
\, \ln \biggl( \frac{p_{Y^n|X^n}(y^n | x^n)}{q_{Y^n|X^n}(y^n | x^n)} \biggr)
\, \mathrm{d}x^n \, \mathrm{d}y^n \nonumber \\
& = \int p_{X^n, Y^n}(x^n, y^n)
\, \ln \Biggl( \frac1n \sum_{i=1}^n \frac{r_{Y|X}(y_i | x_i)}{q_{Y|X}(y_i | x_i)} \Biggr)
\, \mathrm{d}x^n \, \mathrm{d}y^n \nonumber \\
&= \int_0^\infty \frac{1}{u} \, \Biggl[ e^{-u} - \int p_{X^n,Y^n}(x^n,y^n)
\exp \Biggl(-\frac{u}{n} \sum_{i=1}^n
\frac{r_{Y|X}(y_i | x_i)}{q_{Y|X}(y_i | x_i)} \Biggr) \, \mathrm{d}x^n \,
\mathrm{d}y^n \Biggr] \, \mathrm{d}u \\
&= \int_0^\infty \frac{1}{u} \, \Biggl[ e^{-u} - \int
\frac1n \sum_{i=1}^n \Biggl\{ \prod_{j \neq i} q_{Y|X}(y_j | x_j) \, p_X(x_j)
\cdot r_{Y|X}(y_i | x_i) \, p_X(x_i) \Biggr\} \nonumber \\
& \hspace*{3.5cm} \cdot \exp \Biggl(-\frac{u}{n} \sum_{i=1}^n
\frac{r_{Y|X}(y_i | x_i)}{q_{Y|X}(y_i | x_i)} \Biggr) \, \mathrm{d}x^n \,
\mathrm{d}y^n \Biggr] \, \mathrm{d}u \\
&= \int_0^\infty \frac{1}{u} \, \Biggl[ e^{-u} - \int
\frac1n \sum_{i=1}^n \Biggl\{ \prod_{j \neq i} q_{Y|X}(y_j | x_j) \, p_X(x_j) \,
\exp \biggl(-\frac{u}{n} \, \frac{r_{Y|X}(y_j | x_j)}{q_{Y|X}(y_j | x_j)} \biggr) \nonumber \\
& \hspace*{4.5cm} \cdot r_{Y|X}(y_i | x_i) \, p_X(x_i) \exp \biggl(-\frac{u}{n} \,
\frac{r_{Y|X}(y_i | x_i)}{q_{Y|X}(y_i | x_i)} \biggr) \Biggr\} \, \mathrm{d}x^n \,
\mathrm{d}y^n \Biggr] \, \mathrm{d}u \\
&= \int_0^\infty \frac{1}{u} \, \Biggl[ e^{-u} -
\frac1n \sum_{i=1}^n \Biggl\{ \prod_{j \neq i} \int q_{Y|X}(y_j | x_j) \, p_X(x_j) \,
\exp \biggl(-\frac{u}{n} \, \frac{r_{Y|X}(y_j | x_j)}{q_{Y|X}(y_j | x_j)} \biggr) \,
\mathrm{d}x_j \, \mathrm{d}y_j \nonumber \\
& \hspace*{4.5cm} \cdot \int r_{Y|X}(y_i | x_i) \, p_X(x_i) \exp \biggl(-\frac{u}{n} \,
\frac{r_{Y|X}(y_i | x_i)}{q_{Y|X}(y_i | x_i)} \biggr) \, \mathrm{d}x_i \, \mathrm{d}y_i \, \Biggr\}
\Biggr] \, \mathrm{d}u \\
&= \int_0^\infty \frac{1}{u} \, \Biggl[ e^{-u} -
\frac1n \sum_{i=1}^n \Biggl\{ \Biggl( \int q_{Y|X}(y|x) \, p_X(x) \,
\exp \biggl(-\frac{u}{n} \, \frac{r_{Y|X}(y|x)}{q_{Y|X}(y|x)} \biggr) \,
\mathrm{d}x \, \mathrm{d}y \Biggr)^{n-1} \nonumber \\
& \hspace*{4.5cm} \cdot \int r_{Y|X}(y|x) \, p_X(x) \exp \biggl(-\frac{u}{n} \,
\frac{r_{Y|X}(y|x)}{q_{Y|X}(y|x)} \biggr) \, \mathrm{d}x \, \mathrm{d}y \, \Biggr\}
\Biggr] \, \mathrm{d}u \\
&= \int_0^\infty \frac{1}{u} \, \Biggl[ e^{-u} -
\Biggl( \int q_{Y|X}(y|x) \, p_X(x) \,
\exp \biggl(-\frac{u}{n} \, \frac{r_{Y|X}(y|x)}{q_{Y|X}(y|x)} \biggr) \,
\mathrm{d}x \, \mathrm{d}y \Biggr)^{n-1} \nonumber \\
& \hspace*{3cm} \cdot \int r_{Y|X}(y|x) \, p_X(x) \exp \biggl(-\frac{u}{n} \,
\frac{r_{Y|X}(y|x)}{q_{Y|X}(y|x)} \biggr) \, \mathrm{d}x \, \mathrm{d}y \,
\Biggr] \, \mathrm{d}u \\
\label{J1-proof}
&= \int_0^\infty \frac{1}{u} \, \Bigl[ e^{-u} - f^{n-1}\Bigl(\frac{u}{n}\Bigr) \,
g\Bigl(\frac{u}{n}\Bigr) \Bigr] \, \mathrm{d}u,
\end{align}
where $f(\cdot)$ and $g(\cdot)$ are defined in \eqref{f} and \eqref{g}, respectively.
Consequently, $f(0) = g(0) = 1$, and $0 \leq f(u), g(u) \leq 1$ for all $u>0$.

\subsection*{{\em D.2} Proof of Eq.~\eqref{J2}}

\vspace*{-0.8cm}
\begin{align}
& \int p_{X^n, Y^n}(x^n, y^n)
\, \ln q_{Y^n|X^n}(y^n | x^n) \, \mathrm{d}x^n \, \mathrm{d}y^n \nonumber \\
&= \int p_{X^n, Y^n}(x^n, y^n)
\, \sum_{j=1}^n \ln q_{Y|X}(y_j|x_j) \, \mathrm{d}x^n \, \mathrm{d}y^n \\
&= \int \prod_{\ell=1}^n p_X(x_\ell) \cdot
\frac1n \sum_{i=1}^n \Biggl\{ \prod_{\ell \neq i} q_{Y|X}(y_\ell | x_\ell) \, r_{Y|X}(y_i | x_i) \Biggr\}
\sum_{j=1}^n \ln q_{Y|X}(y_j|x_j) \, \mathrm{d}x^n \, \mathrm{d}y^n \\[0.1cm]
&= \int \prod_{\ell=1}^n p_X(x_\ell) \cdot
\frac1n \sum_{i=1}^n \sum_{j=1}^n \Biggl\{ \prod_{\ell \neq i} q_{Y|X}(y_\ell | x_\ell) \cdot r_{Y|X}(y_i | x_i)
\, \ln q_{Y|X}(y_j|x_j) \Biggr\} \, \mathrm{d}x^n \, \mathrm{d}y^n \\[0.1cm]
\label{20200106a1}
&= \frac1n \, \int_{\cX^n} \prod_{\ell=1}^n p_X(x_\ell) \, \Biggl(
\sum_{i=1}^n \sum_{j=1}^n \int_{\cY^n} \prod_{\ell \neq i}
q_{Y|X}(y_\ell | x_\ell) \cdot r_{Y|X}(y_i | x_i) \, \ln q_{Y|X}(y_j|x_j)
\, \mathrm{d}y^n \Biggr) \, \mathrm{d}x^n
\end{align}
We next calculate the inner integral on the right--hand side of \eqref{20200106a1}.
For $i=j$,
\begin{align}
& \int_{\cY^n} \prod_{\ell \neq i}
q_{Y|X}(y_\ell | x_\ell) \cdot r_{Y|X}(y_i | x_i) \, \ln q_{Y|X}(y_j|x_j)
\, \mathrm{d}y^n \nonumber \\
\label{20200106a2}
&= \prod_{\ell \neq i} \int_{\cY} q_{Y|X}(y_\ell | x_\ell) \, \mathrm{d}y_\ell
\cdot \int_{\cY} r_{Y|X}(y_i | x_i) \, \ln q_{Y|X}(y_i|x_i) \, \mathrm{d}y_i \nonumber \\
&= \int_{\cY} r_{Y|X}(y|x_i) \, \ln q_{Y|X}(y|x_i) \, \mathrm{d}y,
\end{align}
else,
\begin{align}
& \int_{\cY^n} \prod_{\ell \neq i}
q_{Y|X}(y_\ell | x_\ell) \cdot r_{Y|X}(y_i | x_i) \, \ln q_{Y|X}(y_j|x_j)
\, \mathrm{d}y^n \nonumber \\
&= \prod_{\ell \notin \{i,j\}} \int_{\cY} q_{Y|X}(y_\ell | x_\ell) \, \mathrm{d}y_\ell
\cdot \int_{\cY} q_{Y|X}(y_j|x_j) \, \ln q_{Y|X}(y_j|x_j) \, \mathrm{d}y_j \cdot
\int_{\cY} r_{Y|X}(y_i | x_i) \, \mathrm{d}y_i \nonumber \\
\label{20200106a3}
&= \int_{\cY} q_{Y|X}(y|x_j) \, \ln q_{Y|X}(y|x_j) \, \mathrm{d}y.
\end{align}
Hence, from \eqref{20200106a1}--\eqref{20200106a3},
\begin{align}
& \int p_{X^n, Y^n}(x^n, y^n)
\, \ln q_{Y^n|X^n}(y^n | x^n) \, \mathrm{d}x^n
\, \mathrm{d}y^n \nonumber \\[0.1cm]
&= \frac1n \, \int_{\cX^n} \prod_{\ell=1}^n p_X(x_\ell) \, \Biggl(
\sum_{i=1}^n \int_{\cY} r_{Y|X}(y|x_i) \, \ln q_{Y|X}(y|x_i) \, \mathrm{d}y \nonumber \\
& \hspace*{3.6cm} + \sum_{i=1}^n \sum_{j \neq i} \int_{\cY} q_{Y|X}(y|x_j) \, \ln q_{Y|X}(y|x_j)
\, \mathrm{d}y \Biggr) \, \mathrm{d}x^n \nonumber \\[0.1cm]
&= \frac1n \Biggl[ \, \sum_{i=1}^n \Biggl\{ \, \prod_{\ell \neq i}
\int_{\cX} p_X(x_\ell) \, \mathrm{d}x_\ell \cdot
\int_{\cX \times \cY} r_{Y|X}(y|x_i) \, \ln q_{Y|X}(y|x_i) \, p_X(x_i) \,
\mathrm{d}x_i \, \mathrm{d}y \Biggr\} \nonumber \\[0.1cm]
&\hspace*{1cm} + \sum_{i=1}^n \sum_{j \neq i} \Biggl\{ \,
\prod_{\ell \neq j} \int_{\cX} p_X(x_\ell) \, \mathrm{d}x_\ell
\cdot \int_{\cX \times \cY} p_X(x_j) \, q_{Y|X}(y|x_j) \, \ln q_{Y|X}(y|x_j) \,
\mathrm{d}x_j \, \mathrm{d}y \Biggr\} \Biggr] \nonumber \\[0.1cm]
&= \frac1n \Biggl[ \, \sum_{i=1}^n
\int_{\cX \times \cY} r_{Y|X}(y|x) \, \ln q_{Y|X}(y|x) \, p_X(x) \,
\mathrm{d}x \, \mathrm{d}y \nonumber \\
& \hspace*{1cm} + \sum_{i=1}^n \sum_{j \neq i}
\int_{\cX \times \cY} p_X(x) \, q_{Y|X}(y|x) \, \ln q_{Y|X}(y|x) \,
\mathrm{d}x \, \mathrm{d}y \Biggr] \nonumber \\[0.1cm]
&= \int_{\cX \times \cY} p_X(x) \, r_{Y|X}(y|x) \, \ln q_{Y|X}(y|x) \,
\mathrm{d}x \, \mathrm{d}y \nonumber \\
& \hspace*{0.3cm} + (n-1)
\int_{\cX \times \cY} p_X(x) \, q_{Y|X}(y|x) \, \ln q_{Y|X}(y|x) \,
\mathrm{d}x \, \mathrm{d}y.
\end{align}

\subsection*{{\em D.3} Proof of Eqs.~\eqref{output dist.}--\eqref{t}}

\vspace*{-0.7cm}
\begin{align}
p_{Y^n}(y^n) &= \int p_{Y^n|X^n}(y^n | x^n)
\, p_{X^n}(x^n) \, \mathrm{d}x^n \nonumber \\[0.1cm]
&= \frac1n \sum_{i=1}^n \Biggl\{ \prod_{j \neq i} \int q_{Y|X}(y_j | x_j)
\, p_X(x_j) \, \mathrm{d}x_j \cdot
\int r_{Y|X}(y_i | x_i) P_X(x_i) \, \mathrm{d}x_i \Biggr\} \nonumber \\[0.1cm]
&= \frac1n \sum_{i=1}^n \biggl\{ \prod_{j \neq i} v(y_j) \cdot w(y_i) \biggr\} \nonumber \\[0.1cm]
\label{20200107a1}
&= \prod_{j=1}^n v(y_j) \cdot \frac1n  \sum_{i=1}^n \frac{w(y_i)}{v(y_i)},
\quad \forall \, y^n \in \cY^n,
\end{align}
where $v(\cdot)$ and $w(\cdot)$ are probability densities on $\cY$, as defined
in \eqref{v} and \eqref{w}, respectively. This proves \eqref{output dist.}.

We next prove \eqref{eq: mean of Z ln Z}, which is used to calculate the
entropy of $Y^n$ with the density $p_{Y^n}(\cdot)$ in \eqref{20200107a1}.
In view of the integral representation of the logarithmic function in \eqref{ir},
and by interchanging the order of the integrations, we get that for a positive random variable $Z$
\begin{align}
\bE\bigl\{Z \ln Z\bigr\}
&= \int_0^\infty \frac1u \cdot \bE \bigl\{ Z \, \bigl(e^{-u} - e^{-uZ}\bigr) \bigr\}
\, \mathrm{d}u \nonumber \\[0.05cm]
&= \int_0^\infty \frac{\bE\bigl\{Z\bigr\} \, e^{-u} - \bE\bigl\{ Z e^{-uZ} \bigr\}}{u}
\; \mathrm{d}u \nonumber \\[0.05cm]
&= \int_0^\infty \frac{M'_Z(0) \, e^{-u} - M'_Z(-u)}{u} \; \mathrm{d}u,
\end{align}
which proves \eqref{eq: mean of Z ln Z}.
Finally, we prove \eqref{entropy of Y vec}.
In view of \eqref{20200107a1},
\begin{align}
h(Y^n)
&= -\int p_{Y^n}(y^n) \, \ln p_{Y^n}(y^n)
\, \mathrm{d}y^n \nonumber \\
&= -\int \prod_{j=1}^n v(y_j) \cdot \frac1n  \sum_{i=1}^n \frac{w(y_i)}{v(y_i)} \cdot
\Biggl[ \, \ln \Biggl( \, \prod_{j=1}^n v(y_j) \Biggr) +
\ln \biggl( \frac1n  \sum_{i=1}^n \frac{w(y_i)}{v(y_i)} \biggr) \Biggl]
\, \mathrm{d}y^n \nonumber \\
&= -\int \prod_{j=1}^n v(y_j) \cdot \frac1n  \sum_{i=1}^n \frac{w(y_i)}{v(y_i)} \cdot
\Biggl[ \, \sum_{j=1}^n \ln v(y_j) +
\ln \biggl( \frac1n  \sum_{i=1}^n \frac{w(y_i)}{v(y_i)} \biggr) \Biggl]
\, \mathrm{d}y^n \nonumber \\
&= -\int \prod_{\ell=1}^n v(y_\ell) \cdot \frac1n \sum_{i=1}^n \sum_{j=1}^n \frac{w(y_i)
\, \ln v(y_j)}{v(y_i)} \; \mathrm{d}y^n \nonumber \\
\label{20200107a2}
&\hspace*{0.4cm} -\int \prod_{j=1}^n v(y_j) \cdot \frac1n  \sum_{i=1}^n \frac{w(y_i)}{v(y_i)}
\cdot \ln \Biggl( \frac1n  \sum_{i=1}^n \frac{w(y_i)}{v(y_i)} \Biggr) \; \mathrm{d}y^n.
\end{align}
A calculation of the first integral on the right--hand side of \eqref{20200107a2} gives
\begin{align}
& \int \prod_{\ell=1}^n v(y_\ell) \cdot \frac1n \sum_{i=1}^n \sum_{j=1}^n \frac{w(y_i)
\, \ln v(y_j)}{v(y_i)} \; \mathrm{d}y^n \nonumber \\
&= \frac1n \sum_{i=1}^n \sum_{j=1}^n \int \prod_{\ell=1}^n v(y_\ell) \cdot
\frac{w(y_i) \, \ln v(y_j)}{v(y_i)} \; \mathrm{d}y^n \nonumber \\
\label{20200107a3}
&= \frac1n \sum_{i=1}^n \sum_{j=1}^n \int \prod_{\ell \neq i} v(y_\ell) \cdot
w(y_i) \, \ln v(y_j) \, \mathrm{d}y^n.
\end{align}
For $i=j$, the inner integral on the
right--hand side of \eqref{20200107a3} satisfies
\begin{align}
& \int \prod_{\ell \neq i} v(y_\ell) \cdot w(y_i) \, \ln v(y_j) \, \mathrm{d}y^n \nonumber \\
&= \prod_{\ell \neq i} \int v(y_\ell) \, \mathrm{d}y_\ell \cdot \int w(y_i) \, \ln v(y_i) \, \mathrm{d}y_i \nonumber \\
\label{20200107a4}
&= \int w(y) \, \ln v(y) \, \mathrm{d}y,
\end{align}
and for $i \neq j$,
\begin{align}
& \int \prod_{\ell \neq i} v(y_\ell) \cdot w(y_i) \, \ln v(y_j) \, \mathrm{d}y^n \nonumber \\
&= \prod_{\ell \neq i, j} \int v(y_\ell) \, \mathrm{d}y_\ell \cdot \int w(y_i) \, \mathrm{d}y_i
\cdot \int v(y_j) \, \ln v(y_j) \, \mathrm{d}y_j \nonumber \\
\label{20200107a5}
&= \int v(y) \, \ln v(y) \, \mathrm{d}y.
\end{align}
Therefore combining \eqref{20200107a3}--\eqref{20200107a5} gives
\begin{align}
& \int \prod_{\ell=1}^n v(y_\ell) \cdot \frac1n \sum_{i=1}^n \sum_{j=1}^n \frac{w(y_i)
\, \ln v(y_j)}{v(y_i)} \; \mathrm{d}y^n \nonumber \\
\label{20200107a6}
&= \int w(y) \, \ln v(y) \, \mathrm{d}y + (n-1) \int v(y) \, \ln v(y) \, \mathrm{d}y.
\end{align}
Finally, we calculate the
second integral on the right--hand side of \eqref{20200107a2}. Let $\mu_n$
be the probability density function defined as
\begin{align}
\mu_n(y^n) \dfn \prod_{j=1}^n v(y_j), \quad y^n \in \cY^n,
\end{align}
and let
\begin{align}
\label{eq: RV Z}
Z \dfn \frac1n \sum_{i=1}^n \frac{w(V_i)}{v(V_i)}
\end{align}
where $\{V_i\}_{i=1}^n$ are i.i.d.\ $\cY$-valued random variables with
a probability density function $v$. Then, in view of \eqref{eq: mean of Z ln Z},
the second integral on the right--hand side of \eqref{20200107a2} satisfies
\begin{align}
& \int \prod_{j=1}^n v(y_j) \cdot \frac1n  \sum_{i=1}^n \frac{w(y_i)}{v(y_i)}
\cdot \ln \Biggl( \frac1n  \sum_{i=1}^n \frac{w(y_i)}{v(y_i)} \Biggr) \; \mathrm{d}y^n \nonumber \\
&= \bE \bigl\{ Z \ln Z \bigr\} \nonumber \\
\label{20200108a1}
&= \int_0^\infty \frac{M'_Z(0) \, e^{-u} - M'_Z(-u)}{u} \; \mathrm{d}u.
\end{align}
The MGF of $Z$ is equal to
\begin{align}
M_Z(u) &= \int_{\cY^n} \prod_{i=1}^n v(r_i) \, \exp\Biggl(\frac{u}{n}
\sum_{i=1}^n \frac{w(r_i)}{v(r_i)}\Biggr) \, \mathrm{d}\underline{r} \nonumber \\
&= \prod_{i=1}^n \int_{\cY} v(r_i) \, \exp \biggl( \frac{u}{n} \, \frac{w(r_i)}{v(r_i)}
\biggr) \, \mathrm{d}r_i \nonumber \\
\label{eq: M_Z}
&= K^n \Bigl(\frac{u}{n}\Bigr),
\end{align}
where
\begin{align}
\label{eq: K}
K(u) \dfn \int_{\cY} v(y) \, \exp\biggl(\frac{u \, w(y)}{v(y)}\biggr) \, \mathrm{d}y, \quad \forall \, u \in \reals,
\end{align}
and consequently, \eqref{eq: K} yields
\begin{align}
\label{eq: diff M_Z}
M'_Z(u) &= K^{n-1}\Bigl(\frac{u}{n}\Bigr) \, K'\Bigl(\frac{u}{n}\Bigr) \nonumber \\
&= \Biggl( \int v(y) \, \exp\biggl(\frac{u \, w(y)}{v(y)}\biggr) \, \mathrm{d}y \Biggr)^{n-1}
\, \int w(y) \, \exp\biggl(\frac{u \, w(y)}{v(y)}\biggr) \, \mathrm{d}y,
\end{align}
and
\begin{align}
\label{eq: diff M_Z at zero}
M'_Z(0) = 1.
\end{align}
Therefore, combining \eqref{20200108a1}--\eqref{eq: diff M_Z at zero}
gives the following single-letter expression for the second multi-dimensional
integral on the right--hand side of \eqref{20200107a2}:
\begin{align}
\label{20200108a2}
\int \prod_{j=1}^n v(y_j) \cdot \frac1n  \sum_{i=1}^n \frac{w(y_i)}{v(y_i)}
\cdot \ln \Biggl( \frac1n  \sum_{i=1}^n \frac{w(y_i)}{v(y_i)} \Biggr) \; \mathrm{d}y^n
= \int_0^\infty \frac1u \, \biggl[ e^{-u} - t^{n-1}\Bigl(\frac{u}{n}\Bigr) \, s\Bigl(\frac{u}{n}\Bigr) \biggr] \, \mathrm{d}u,
\end{align}
where the functions $s(\cdot)$ and $t(\cdot)$ are defined in \eqref{s} and \eqref{t}, respectively.
Combining \eqref{20200107a2}, \eqref{20200107a6} and \eqref{20200108a2} gives \eqref{entropy of Y vec}.

\subsection*{{\em D.4} Specialization to a BSC with Jamming}

In the BSC example considered, we have $\cX = \cY = \{0, 1\}$, and
\begin{align}
\label{BSC p}
& r_{Y|X}(y|x) = \varepsilon \, 1\{x \neq y\} + (1-\varepsilon) \, 1\{x=y\}, \\
\label{BSC q}
& q_{Y|X}(y|x) = \delta \, 1\{x \neq y\} + (1-\delta) \, 1\{x=y\},
\end{align}
where $1\{\mathrm{relation}\}$ is the indicator function that is equal to~1 if the relation
holds, and to zero otherwise. Recall that we assume $0 < \delta < \varepsilon \leq \tfrac12$.
Let
\begin{align}
\label{symmetric binary input}
p_X(0) = p_X(1) = \tfrac12,
\end{align}
be the binary symmetric source (BSS). From \eqref{f} and \eqref{g}, for $u \geq 0$,
\begin{align}
f(u) &= \sum_{x,y} p_X(x) \, q_{Y|X}(y|x) \, \exp\biggl(-\frac{u \, r_{Y|X}(y|x)}{q_{Y|X}(y|x)}\biggr) \nonumber \\
\label{20200109a1}
&= (1-\delta) \, \exp\biggl(-\frac{(1-\varepsilon) u}{1-\delta}\biggr) + \delta \, \exp\biggl(-\frac{\varepsilon u}{\delta}\biggr), \\
g(u) &= \sum_{x,y} p_X(x) \, r_{Y|X}(y|x) \, \exp\biggl(-\frac{u \, r_{Y|X}(y|x)}{q_{Y|X}(y|x)}\biggr) \nonumber \\
\label{20200109a2}
&= (1-\varepsilon) \, \exp\biggl(-\frac{(1-\varepsilon) u}{1-\delta}\biggr) + \varepsilon \, \exp\biggl(-\frac{\varepsilon u}{\delta}\biggr).
\end{align}
Furthermore, we get from \eqref{BSC p}, \eqref{BSC q} and \eqref{symmetric binary input} that
\begin{align}
- \sum_{x,y} p_X(x) \, r_{Y|X}(y|x) \, \ln q_{Y|X}(y|x)
&= -\varepsilon \ln \delta - (1-\varepsilon) \ln (1-\delta) \nonumber \\
\label{20200109a3}
&= d(\varepsilon \| \delta) + H_{\mathrm{b}}(\varepsilon),
\end{align}
and
\begin{align}
\label{20200109a4}
& - \sum_{x,y} p_X(x) \, q_{Y|X}(y|x) \, \ln q_{Y|X}(y|x) = H_{\mathrm{b}}(\delta).
\end{align}
Substituting \eqref{20200109a1}--\eqref{20200109a4} into \eqref{20200108a3}
(where integrals in \eqref{20200108a3} are replaced by sums) gives
\begin{align}
\label{20200109a5}
H(Y^n | X^n) &= d(\varepsilon \| \delta) + H_{\mathrm{b}}(\varepsilon)
+ (n-1) \, H_{\mathrm{b}}(\delta)
+ \int_0^\infty \Bigl[ f^{n-1}\Bigl(\frac{u}{n}\Bigr) \, g\Bigl(\frac{u}{n}\Bigr) - e^{-u} \Bigr]
\, \frac{\mathrm{d}u}{u}.
\end{align}

Since the input is a BSS, due to the symmetry of the channel
\eqref{transition law - jamming},
the output is also a BSS. This implies that (in units of nats)
\begin{align}
\label{20200109a6}
H(Y^n) = n \ln 2.
\end{align}
As a sanity check, we verify it by using \eqref{entropy of Y vec}.
From \eqref{v} and \eqref{w}, for $y \in \{0,1\}$,
\begin{align}
\label{v - BSC}
& v(y) = p_X(0) \, q_{Y|X}(y|0) + p_X(1) \, q_{Y|X}(y|1) = \tfrac12, \\
\label{w - BSC}
& w(y) = p_X(0) \, r_{Y|X}(y|0) + p_X(1) \, r_{Y|X}(y|1) = \tfrac12,
\end{align}
and, from \eqref{s} and \eqref{t}, it consequently follows that
\begin{align}
\label{s - BSC}
s(u) = w(0) \, \exp\biggl(-\frac{u \, w(0)}{v(0)}\biggr) + w(1) \, \exp\biggl(-\frac{u \, w(1)}{v(1)}\biggr)
= e^{-u}, \quad \forall \, u \geq 0,
\end{align}
and also
\begin{align}
\label{t - BSC}
t(u) = e^{-u}, \quad \forall \, u \geq 0.
\end{align}
It can be verified that substituting \eqref{v - BSC}--\eqref{t - BSC} into
\eqref{entropy of Y vec} reproduces \eqref{20200109a6}. Finally, subtracting
\eqref{20200109a5} from \eqref{20200109a6} gives \eqref{eq: MI BSC and jamming}.


\begin{thebibliography}{10}

\bibitem{Alzaatreh16}
A. Alzaatreh, C. Lee, F. Famoye, I. Ghosh, ``The generalized Cauchy family
of distributions with applications,'' {\em J. Stat. Distrib. App.}, vol.~3, paper~12,
2016. 

\bibitem{Arikan96}
E. Arikan, ``An inequality on guessing and its application to
sequential decoding,'' {\em IEEE Transactions on Information Theory},
vol.~42, no.~1, pp.~99--105, January 1996.

\bibitem{ArikanM98-1}
E. Arikan and N. Merhav, ``Guessing subject to distortion,'' {\em IEEE Transactions on
Information Theory}, vol.~44, no.~3, pp.~1041--1056, May 1998.

\bibitem{ArikanM98-2}
E. Arikan and N. Merhav, ``Joint source-channel coding and guessing
with application to sequential decoding,'' {\em IEEE Transactions on Information
Theory}, vol.~44, no.~5, pp.~1756--1769, September 1998.

\bibitem{BK_ECP13}
D. Berend and A. Kontorovich, ``On the concentration of the missing mass,''
{\em Electronic Communications in Probability}, vol.~18, paper~3, pp.~1--7,
2013.

\bibitem{Boztas97}
S. Bozta\c{s}, ``Comments on ``An inequality on guessing and its application
to sequential decoding'','' {\em IEEE Transactions on Information Theory}, vol.~43,
no.~6, pp.~2062--2063, November 1997.

\bibitem{BracherHL_IT19}
A. Bracher, E. Hof and A. Lapidoth, ``Guessing attacks on distributed--storage
systems,'' {\em IEEE Transactions on Information Theory}, vol.~65, no.~11, pp.~6975--6998,
November 2019.

\bibitem{BunteL14a}
C. Bunte and A. Lapidoth, ``Encoding tasks and R\'{e}nyi entropy,'' {\em IEEE
Trans. on Information Theory}, vol.~60, no.~9, pp.~5065--5076, September 2014.

\bibitem{Campbell65}
L. L. Campbell, ``A coding theorem and R\'{e}nyi's entropy,''
{\em Information and Control}, vol.~8, no.~4, pp.~423--429, August 1965.

\bibitem{Carrillo10}
R. E. Carrillo, T. C. Aysal and K. E. Barner, ``A generalized Cauchy distribution
framework for problems requiring robust behavior,'' {\em EURASIP J. Adv. Signal
Process.}, 2010.

\bibitem{CV2014a}
T. Courtade and S. Verd\'{u}, ``Cumulant generating function of codeword
lengths in optimal lossless compression,'' {\em Proceedings of the 2014
IEEE International Symposium on Information Theory}, pp.~2494--2498, Honolulu,
Hawaii, USA, July 2014.

\bibitem{CV2014b}
T. Courtade and S. Verd\'{u}, ``Variable-length lossy compression and channel coding:
Non-asymptotic converses via cumulant generating functions,'' {\em Proceedings of the 2014
IEEE International Symposium on Information Theory}, pp.~2499--2503, Honolulu,
Hawaii, USA, July 2014.

\bibitem{DZWY15}
A.~Dong, H.~Zhang, D.~Wu, and D.~Yuan, ``Logarithmic expectation of the
sum of exponential random variables for wireless communication performance
evaluation,'' {\it Proc.\ 2015 IEEE 82nd Vehicular Technology Conference},
Boston, MA, USA, September 2015.

\bibitem{EN93}
S.~E.~Esipov and T.~J.~Newman, ``Interface growth and Burgers
turbulence: the problem of random initial conditions,''
{\em Phys. Rev.~E}, vol.~48, no.~2, pp.~1046--1050, August 1993.

\bibitem{EvansB88}
R.~J.~Evans, J.~Boersma, N.~M.~Blachman and A.~A.~Jagers, ``The entropy of a
Poisson distribution,'' {\em SIAM Review}, vol.~30, no.~2, pp.~314--317,
June 1988.

\bibitem{GR14}
I.~S.~Gradshteyn and I.~M.~Ryzhik, {\em Tables of Integrals, Series, and
Products}, Eighth Edition, Elsevier, 2014.

\bibitem{Gzyl1}
H. Gzyl and A. Tagliani, ``Stieltjes moment problem and fractional moments,''
{\em Applied Mathematics and Computation}, vol.~216, no.~11, pp.~3307--3318,
August 2010.

\bibitem{Gzyl2}
H. Gzyl and A. Tagliani, ``Determination of the distribution of total loss
from the fractional moments of its exponential,'' {\em Applied Mathematics
and Computation}, vol.~219, no.~4, pp.~2124--2133, November 2012.

\bibitem{HanawalS11}
M. K. Hanawal and R. Sundaresan, ``Guessing revisited: a large deviations approach,''
{\em IEEE Transactions on Information Theory}, vol.~57, no.~1, pp.~70--78, January 2011.

\bibitem{Knessl98}
C. Knessl, ``Integral representations and asymptotic expansions for
Shannon and R\'{e}nyi entropies,'' {\em Applied Mathematical Letters},
vol.~11, no.~2, pp.~69--74, 1998.

\bibitem{LapidothM-IT03}
A. Lapidoth and S. Moser, ``Capacity bounds via duality with applications
to multiple-antenna systems on flat fading channels,'' {\em IEEE Transactions
on Information Theory}, vol. 49, no. 10, pp. 2426–2467, October 2003.

\bibitem{MK18}
A. Marsiglietti and V. A. Kostina, ``lower bound on the differential entropy
of log-concave random vectors with applications,'' {\em Entropy}, vol. 20, 
paper~185, 2018. 

\bibitem{Martinez07}
A.~Martinez, ``Spectral efficiency of optical direct detection,''
{\em Journal of the Optical Society of America B}, vol.~24, no.~4,
pp.~739--749, April 2007.

\bibitem{MM09}
M.~M\'ezard and A.~Montanari, {\em Information, Physics, and Computation},
Oxford University Press, New-York, USA, 2009.

\bibitem{NM18_estimation}
N.~Merhav, ``Lower bounds on exponential moments of the quadratic error
in parameter estimation,'' {\em IEEE Transactions on Information Theory},
vol.~64, no.~12, pp.~7636--7648, December 2018.

\bibitem{MC20}
N.~Merhav and A.~Cohen, ``Universal randomized guessing with application to
asynchronous decentralized brute-force attacks,'' {\em IEEE Transactions on
Information Theory}, vol.~66, no.~1, pp.~114--129, January 2020.

\bibitem{MS_Ent20}
N.~Merhav and I.~Sason, ``An integral representation of the logarithmic function
with applications in information theory,'' {\em Entropy}, vol.~22, no.~1, paper~51,
pp.~1--22, January 2020.

\bibitem{OlverLBC10}
F. W. J. Olver, D. W. Lozier, R. F. Boisvert and C. W. Clark, {\em NIST Handbook
of Mathematical Functions}, NIST (National Institute of Standards and Technology)
and Cambridge University Press, New York, USA, 2010.

\bibitem{MRIS_FnT19}
M. Raginsky and I. Sason, {\em Concentration of Measure Inequalities in
Information Theory, Communications and Coding: Third Edition}, pp.~1--261,
Foundations and Trends in Communications and Information Theory,
NOW Publishers, Delft, 2019.

\bibitem{RajanT15}
A. Rajan and C. Tepedelenlio\v{g}lu, ``Stochastic ordering of fading channels
through the Shannon transform,'' {\em IEEE Transactions on Information Theory}, vol.~61,
no.~4, pp.~1619--1628, April 2015.

\bibitem{Renyientropy}
A. R\'{e}nyi, ``On measures of entropy and information,'' {\em Proceedings
of the Fourth Berkeley Symposium on Probability Theory and Mathematical Statistics},
pp.~547--561, Berkeley, California, USA, 1961.

\bibitem{MM19}
S. Salamatian, W. Huleihel, A. Beirami, A. Cohen and M. M\'{e}dard,
``Why botnets work: distributed brute-force attacks need no synchronization,''
{\em IEEE Transactions on Information Forensics and Security}, vol.~14, no.~9,
pp.~2288--2299, September 2019.

\bibitem{ISSV_IT18}
I. Sason and S. Verd\'{u}, ``Improved bounds on lossless source coding and guessing
moments via R\'{e}nyi measures,'' {\em IEEE Transactions on Information Theory},
vol.~64, no.~6, pp.~4323--4346, June 2018.

\bibitem{IS_Ent18}
I. Sason, ``Tight bounds on the R\'{e}nyi entropy via majorization with applications to
guessing and compression,'' {\em Entropy}, vol.~20, no.~12, paper~896,
pp.~1--25, November 2018.

\bibitem{SSRCM19}
J.~Song, S.~Still, R.~D.~H.~Rojas, I.~P.~Castillo, and M.~Marsili,
``Optimal work extraction and mutual information in a generalized
Szil\'ard engine,'' arXiv:1910.0419v1, October~9, 2019.

\bibitem{Sundaresan07}
R. Sundaresan, ``Guessing under source uncertainty,'' {\em IEEE Transactions on
Information Theory}, vol.~53, no.~1, pp.~269--287, January 2007.

\bibitem{Sundaresan07b}
R. Sundaresan, ``Guessing based on length functions,'' {\em Proceedings of the 2007
IEEE International Symposium on Information Theory}, pp.~716--719, Nice, France,
June 2007.

\bibitem{Tagliani1}
A. Tagliani, ``On the proximity of distributions in terms of coinciding fractional
moments,'' {\em Applied Mathematics and Computation}, vol.~145, no.~2--3, pp.~501--509,
December 2003.

\bibitem{Tagliani2}
A. Tagliani, ``Hausdorff moment problem and fractional moments: a simplified procedure,''
{\em Applied Mathematics and Computation}, vol.~218, no.~8, pp.~4423--4432, December 2011.

\bibitem{ZadehH16}
P. H. Zadeh and R. Hosseini, ``Expected logarithm of central quadratic form and
its use in KL-divergence of some distributions,''  {\em Entropy}, vol.~18,
no.~8, paper~288, pp.~1--25, August 2016.


\end{thebibliography}
\end{document}